\documentclass[10pt,twocolumn,twoside]{IEEEtran}

\IEEEoverridecommandlockouts
\usepackage{amsmath,mathrsfs, enumerate}
\usepackage{amssymb}
\usepackage{mathtools}
\usepackage{pstricks,pst-plot,psfrag}
\usepackage{hyperref}
\usepackage[all]{xy}
\usepackage{graphicx}
\usepackage{subfigure,xspace,bm}
\usepackage{color}

\usepackage{subfigure}

\usepackage{algorithm}
\usepackage[noend]{algorithmic}
\algsetup{indent=2em}

\usepackage{epstopdf}         
\graphicspath{{./graphics/}} 
\usepackage{stackengine}
 
\allowdisplaybreaks[4]
 
\newtheorem{theorem}{Theorem}[section]

\newtheorem{lemma}[theorem]{Lemma}

\newtheorem{remark}[theorem]{Remark}

\newtheorem{proposition}[theorem]{Proposition}  

\newtheorem{innermodel}{Model}
\newenvironment{model}[1]
  {\renewcommand\theinnermodel{#1}\innermodel}
  {\endinnermodel}

\newcommand\bah[1]{\textcolor{black}{#1}}

\newcommand{\integers}{\mathbb{Z}}
\newcommand{\integerspositive}{\mathbb{Z}_{\geq 1}}

\newcommand{\oprocendsymbol}{\hbox{$\bullet$}}
\newcommand{\oprocend}{\relax\ifmmode\else\unskip\hfill\fi\oprocendsymbol}

\DeclareMathOperator*{\argmin}{arg\,min}
\DeclareMathOperator*{\argmax}{arg\,max}

\newcommand{\longthmtitle}[1]{\mbox{}\textup{\bf (#1):}}

\newcommand{\bZ}[1]{\bar{Z}_{#1}}
\newcommand{\Zprocess}[1]{\{Z_{#1}\}_{n=1}^{\infty}}

\newcommand{\delp}{\delta'}
\newcommand{\dels}{\delta^{\star}}
\newcommand{\dhat}{\hat{\delta}}

\newcommand{\Nsum}{\sum_{i=1}^N}
\newcommand{\tsum}{\sum_{t=1}^n}

\newcommand{\bR}{\bar{R}}
\newcommand{\bB}{\bar{B}}
\newcommand{\bT}{\bar{T}}
\newcommand{\Betafun}{\mathsf{Beta}}
\newcommand{\Polya}{\mathsf{Polya}}

\newcommand{\F}{\mathcal{F}}
\newcommand{\N}{\mathcal{N}}
\newcommand{\G}{{\mathcal{G}}}
\newcommand{\E}{{\mathcal{E}}}

\newcommand{\dsis}{\delta_{SIS}}
\newcommand{\bsis}{\beta_{SIS}}

\renewcommand{\theenumi}{(\roman{enumi})}

\renewcommand{\baselinestretch}{1}

\begin{document}

\title{A Polya Contagion Model for Networks}

\author{Mikhail Hayhoe, Fady Alajaji, \textit{Senior Member, IEEE}, and Bahman Gharesifard, \textit{Member, IEEE}
\thanks{
Mikhail Hayhoe is with the Department of Electrical and Systems Engineering at the University of Pennsylvania, Philadelphia, PA, USA \texttt{mhayhoe@seas.upenn.edu}.

Fady Alajaji and Bahman Gharesifard are with the Department of Mathematics and Statistics at Queen's University, Kingston, ON, Canada,~\texttt{\{fa, ghareb\}@queensu.ca}.

This work was partially supported by the Natural Sciences and Engineering Research Council of Canada. Parts of this work were presented at the 2017 American Control Conference~\cite{MH-FA-BG:17}.}
}
\maketitle

\begin{abstract}
A network epidemics model based on the classical Polya urn scheme is investigated. Temporal contagion processes are generated on the network nodes using a modified Polya sampling scheme that accounts for spatial infection among neighbouring nodes. The stochastic properties and the asymptotic behaviour of the resulting network contagion process are analyzed. Unlike the classical Polya process, the network process is noted to be non-stationary in general, although it is shown to be time-invariant in its first and some of its second-order statistics and to satisfy martingale convergence properties under certain conditions. Three classical Polya processes, one computational and two analytical, are proposed to statistically approximate the contagion process of each node, showing a good fit for a range of system parameters. \bah{Finally, empirical results compare and contrast our model with the well-known discrete time SIS model.}
\\

\end{abstract}

\section{Introduction}\label{sec:introduction}

In this paper we examine the dynamics and properties of a contagion process, or $\emph{epidemic}$, on a network. Here an epidemic can represent a disease~\cite{LK-MA-KD-SK-AO:14}, a computer virus~\cite{MG-WG-DT:03}, the spread of an innovation, rumour or idea~\cite{ER:03}, or the dynamics of competing opinions in a social network~\cite{EA-LAA:05}.

Many different models for the study of infection propagation and curing exist in the literature. Our model\bah{, the network Polya contagion process, bears similarities to} the well-known susceptible-infected-susceptible \bah{(SIS)} infection model~\cite{DE-JK:10}. In this model, all nodes may initially be healthy or infected. As the epidemic spreads, nodes that are infected can be cured to become healthy, but any healthy node may become infected at any time, regardless of whether they have been cured previously. \bah{Epidemics on networks have been intensively studied in recent years; see~\cite{PVM-JO-RK:09} and references therein and thereafter. The model that we present is an adaptation of the classical Polya contagion process~\cite{GP-FE:23,GP-FE:28,GP:30} to a network setting by accounting for spatial infection between nodes. The classical Polya model has been used to study a variety of epidemics such as the bubonic plague in Peru~\cite{AR:40} and the spread of chlamydia in a closed population~\cite{CM-LA-MS:96}, as well as a wide range of other applications; see ~\cite{RP:07} for a survey.} In this work we will examine the stochastic evolution of the network Polya contagion process. 

Our model is motivated by the classical Polya contagion process, which evolves by sampling from an urn containing a finite number of red and black balls~\cite{GP-FE:23,GP-FE:28,GP:30}. In the network Polya contagion model, each node of the underlying network is equipped with an individual urn; however, instead of sampling from these urns when generating its contagion process, each node has a ``super urn'', created by combining the contents of its own urn with those of its neighbours' urns. This adaptation captures the concept of spatial infection, since having infected neighbours increases the chance that an individual is infected in the future. This concept of the super urn sampling mechanism for incorporating spatial interactions was originally introduced in~\cite{AB-PB-FA:99} in the context of the image segmentation and labeling problem. We herein adapt the image model of~\cite{AB-PB-FA:99} for a network setting and analyze the resulting contagion process affecting each node of the network.         

More specifically, we study the time evolution and stochastic properties of the proposed network contagion process. We derive an expression for the temporal $n$-fold joint probability distribution of the process. We show that this process, unlike the classical Polya urn process, is in general non-stationary, and hence not exchangeable. For the special case of complete networks, we analytically find the $1$-dimensional and $2$-dimensional $(n,1)$-step marginal distributions of the contagion process. These results show that, even though it is not stationary, the process \bah{in this case} is nevertheless identically distributed with its later two marginal distributions being invariant to time shifts. We also establish several martingale properties regarding the network urn compositions, proving that the proportions of red balls in each node's urn as well as the network average urn proportion converge almost surely to a limit as time grows without bound. We next provide three approximations to the network contagion process by modelling each node's contagion process via a classical stationary Polya process~\cite{GP:30}. In the first one, we approximate each node's process with the classical Polya process whose correlation parameter is empirically selected so that the Kullback-Leibler divergence measure between its $n$-fold joint distribution and that of the original node process is minimized. In the second approximation, we propose an analytical model whose parameters are chosen by matching its first and $(n,1)$-step second-order statistics with those of the original node process, which fits well for large networks. The last approximation uses a classical Polya model with parameters chosen analytically that we show fits well for small networks. Finally, simulation results are presented to support the validity of these approximations \bah{and to compare our model with the traditional discrete time SIS model, which suggests that the network Polya contagion process captures certain properties of the SIS model, while offering new insights in the case of widespread infection.}

The rest of the paper is organized as follows. Section~\ref{sec:preliminaries} outlines some preliminary knowledge that will be used throughout the paper. Section~\ref{sec:model} introduces the network contagion process, and Section ~\ref{sec:stochastic_properties} presents its stochastic properties and asymptotic behaviour. Section~\ref{sec:model_approximations}  proposes three approximations for the individual node contagion processes in the network, along with numerical modelling results. Lastly, Section~\ref{sec:conclusion} concludes the paper.

\section{Preliminaries}\label{sec:preliminaries}

For a sequence $v_i = (v_{i,1},...,v_{i,n})$, we use the notation $v_{i,s}^t$ with $1 \leq s < t \leq n$ to denote the vector $(v_{i,s},v_{i,s+1},...,v_{i,t})$. Our technical results rely on notions from stochastic processes, some of which we recall here. Throughout, we assume that the reader is familiar with basic notions of probability theory. 

Let $ (\Omega,\F,P) $ be a probability space, and consider the stochastic process $\Zprocess{n}$, where each
$ Z_n $ is a random variable on $\Omega$. We often refer to the indices of the process as ``time'' indices. We recall that the process $\Zprocess{n}$ is \emph{stationary} if for any $n\in \integers_{\geq 1}$, its $n$-fold joint probability distribution (i.e., the distribution of $(Z_1,...,Z_n)$) is invariant to time shifts. Further, $\Zprocess{n}$ is \emph{exchangeable} if for any $n\in \integers_{\geq 1}$, its $n$-fold joint distribution is invariant to permutations of the indices $1,...,n$. It directly follows from the definitions that an exchangeable process is stationary. Lastly, the process $\Zprocess{n}$ is called a \emph{martingale} (resp. \emph{supermartingale}, \emph{submartingale}) with respect to the process $ \{Y_n\}_{n=1}^{\infty} $ if $E[|Z_n|] < \infty$ and $E[Z_{n+1} | Y_{n}] = Z_{n}$ almost surely (resp. less than or equal to, greater than or equal to), for all $n$. Precise definitions of all notions, including that of {\em ergodicity}, can be found in standard texts (e.g.,~\cite{RA-CD:00,GG-DS:01}).

We now recall the classical version of the Polya contagion process~\cite{GP-FE:23,GP:30}. Consider an urn with $R \in \integers_{> 0} $ red balls and $B \in \integers_{> 0}$ black balls. We denote the total number of balls by $ T $, i.e., $T = R + B$. At each time step, a ball is drawn from the urn. The ball is then returned along with $\Delta\ > 0$ balls of the same colour. We use an indicator $Z_n$ to denote the colour of ball in the $n$th draw:
\begin{align*}
	Z_n = \begin{cases}
		1 &\text{if the $n$th draw is red}\\
		0 &\text{if the $n$th draw is black.}
	\end{cases}
\end{align*}
Let $U_n$ denote the proportion of red balls in the urn after the $n$th draw. Then
\begin{align*}
	U_n &:= \frac{R + \Delta\tsum Z_{t}}{T + n\Delta} = \frac{\rho_c + \delta_c\tsum Z_{t}}{1+n\delta_c},
\end{align*}
where $\rho_c = \frac{R}{T}$ is the initial proportion of red balls in the urn and $\delta_c = \frac{\Delta}{T}$ is a correlation parameter. Since we draw balls from this urn at each time step, the conditional probability of drawing a red ball at time $n$, given $Z^{n-1}=(Z_1,\cdots,Z_{n-1})$, is given by
\begin{align*}
	P(Z_n = 1 \ | \ Z^{n-1} ) &= \frac{R + \Delta\sum_{t=1}^{n-1} Z_{t}}{T + (n-1)\Delta} = U_{n-1}.
\end{align*}
It can be easily shown that $\{U_n\}_{n=1}^{\infty}$ is a martingale~\cite{WF:71}.
The process $\{Z_n\}_{n=1}^{\infty}$, whose $n$-fold joint distribution is denoted by $Q_{\rho_c,\delta_c}^{(n)}$, is also exchangeable (hence stationary) and non-ergodic with both $U_n$ and the process sample average $\frac{1}{n}\sum_{i=1}^n Z_i$ converging almost surely as $n \rightarrow \infty$ to a random variable governed by the Beta distribution with parameters $\frac{\rho_c}{\delta_c}$ and $\frac{1-\rho_c}{\delta_c}$; we denote this probability density function (pdf) by $\Betafun(\frac{\rho_c}{\delta_c},\frac{1-\rho_c}{\delta_c})$~\cite{WF:71,FA-TF:94}. Lastly, the $1$-dimensional distribution of the Polya process is
$Q_{\rho_c,\delta_c}^{(1)}(a)=P(Z_n = a) = (\rho_c)^a (1-\rho_c)^{1-a}$, for all $n\in \integers_{\geq 1}$ and $a \in \{0,1\}$. The above classical Polya process $\{Z_n\}_{n=1}^{\infty}$ is fully described by its parameters $\rho_c$ and $\delta_c$, and thus we denote it by $\Polya(\rho_c,\delta_c)$.

\section{Network Polya Contagion Process}\label{sec:model}

In this section, we introduce a generalization of the Polya contagion process to networks, where each individual node in the underlying graph that describes the network topology is still equipped with an urn; however, the node's neighbouring structure affects the evolution of its process. This model hence captures spatial contagion, since infected neighbours increase the chance of a node being infected in the future.

Consider an undirected graph $\G = (V,\E) $, where $V=\{1,\ldots, N\} $ is the set of $ N \in \integers_{\geq 1} $ nodes and $ \E \subset V\times V $ is the set of edges. We assume that $ \G $ is connected, i.e., there is a path between any two nodes in $ \G $. We use $\N_i$ to denote the set of nodes that are neighbours to node $i$, that is $\N_i = \{v \in V : (i,v) \in \E\}$, and $\N_i' = \{i\} \cup \N_i$. If $\N_i' = V $ for all $i \in V$, the network is called complete; if $ |\N_i| = |\N_j| $ for all $i,j\in V$, we call it regular. Each node $i\in V$ is equipped with an urn,  initially with $R_i \in \integers_{> 0}$ red balls and $B_i \in \integers_{> 0}$ black balls (we do not let $R_i = 0$ or $B_i=0$ to avoid any degenerate cases). We let $T_i = R_i + B_i$ be the total number of balls in the $ i $th urn, $i \in \{1,\cdots,N\}$. We use $Z_{i,n}$ as an indicator for the ball drawn for node $i$ at time $n$:
\begin{align*}
	Z_{i,n} = \begin{cases}
		1 &\text{if the $n$th draw for node $i$ is red}\\
		0 &\text{if the $n$th draw for node $i$ is black.}
	\end{cases}
\end{align*}
However, instead of drawing solely from its own urn, each node \bah{draws simultaneously from} a ``super urn'' created by combining all the balls in its own urn with the balls in its neighbours' urns; see Figure~\ref{fig:super_urn}. This allows the spatial relationships between nodes to influence their state. This means that $Z_{i,n}$ is the indicator for a ball drawn from node $i$'s super urn, and not its individual urn. Hence, the super urn of node $ i $ initially has $\bR_i = \sum_{j \in \N_i'}R_j$ red balls, $\bB_i = \sum_{j \in \N_i'}B_j$ black balls, and $\bT_i = \sum_{j \in \N_i'}T_j$ balls in total.
{
\psfrag{1}[rr][rr]{{\tiny $1$}}
\psfrag{2}[rr][rr]{{\tiny $2$}}
\psfrag{3}[rr][rr]{{\tiny $3$}}
\psfrag{4}[rr][rr]{{\tiny $4$}}
\psfrag{Node 1's super urn}[rr][rr]{{\tiny Node 1's super urn}}
\begin{figure}[!ht]
\centering 
\includegraphics[width=0.75\linewidth]{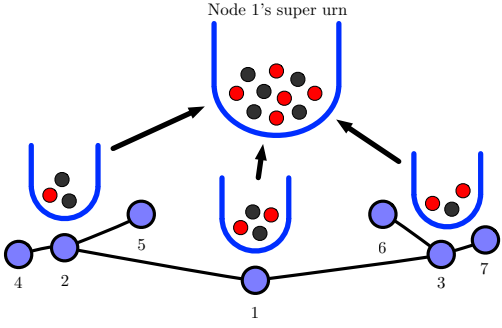}
\caption{\bah{Illustration of a super urn in a network.}}
\label{fig:super_urn}
\end{figure}
}

We further consider a time-varying version of the classical Polya contagion process, following~\cite{RP:90}, where at time $t$ for node $i\in V$, after a red ball is drawn it is returned along with $\Delta_{r,i}(t)$ red balls to node $i$'s urn, and $\Delta_{b,i}(t)$ black balls along with the drawn ball are added to node $i$'s urn when a black ball is drawn. When $\Delta_{r,i}(t) = \Delta_{b,i}(t)$ for all $t \in \integerspositive$, we write $\Delta_i(t)$ instead; if the $\Delta$'s are not node-dependent, we omit the node index. We assume throughout that $\Delta_{r,i}(t) \geq 0, \Delta_{b,i}(t) \geq 0$, for all $t \in \integerspositive $ and that there exists $i\in V$ and $ t $ such that $\Delta_{r,i}(t) + \Delta_{b,i}(t) \neq 0$; otherwise we are simply sampling with replacement.

\bah{In the context of epidemics, the red and black balls in an urn, respectively, represent units of ``infection'' and ``healthiness''; for example, bacteria and white blood cells. In a super urn, the bacteria can infect others in the area and the white blood cells contribute to the overall health in the neighbourhood of an individual. Drawing red at time $ t $ means the bacteria in the neighbourhood were successful in reproduction and so the individual was more infected, otherwise they were healthier since the white blood cells reproduced. Thus when $Z_{i,n}=1$, we declare that node $i$ is {\em infected} at time $n$, and if $Z_{i,n}=0$, then it is {\em healthy}. We add more units of bacteria once they reproduce, but commonly assume this number, $ \Delta_{r,i}(t) $, is the same across all individuals and time because the bacteria does not evolve or become altered. The amount of white blood cells created, $ \Delta_{b,i}(t) $, may change since we can give more medicine to certain people to increase their immune response, or vaccinate them so they are better able to fight the disease.}

To express the proportion of red balls in the individual urns of the nodes, we define the random vector $U_n=(U_{1,n},\ldots, U_{N,n})$, where $U_{i,n}$ is the proportion of red balls in node $i$'s urn after the $n$th draw, $ i \in V $. For node $i$, 
\begin{align*}
	U_{i,n} &:= \frac{R_i + \tsum Z_{i,t}\Delta_{r,i}(t)}{T_i + \tsum Z_{i,t}\Delta_{r,i}(t) + (1-Z_{i,t})\Delta_{b,i}(t)},
\end{align*}
where the numerator represents the total number of red balls in node $i$'s urn after the $n$th draw, while the denominator is the total number of balls in the same urn. Note that $ U_{i,0} = \frac{R_i}{T_i} $ is the initial proportion of balls in node $i$'s urn.
For ease of notation, let
\begin{align}\label{eq:X_n}
X_{j,n}=T_j + \sum_{t=1}^{n} Z_{j,t}\Delta_{r,j}(t) + (1-Z_{j,t})\Delta_{b,j}(t).
\end{align}
Furthermore, we define the random vector $ S_n = (S_{1,n},...,S_{N,n}) $ as the proportion of red balls in the super urns of the nodes after the $ n $th draw, so that $ S_{i,n} $ is the proportion of red balls in node $ i $'s super urn after $ n $ draws. Hence, for node $ i $,
\begin{align}\label{eq:S_n}
	S_{i,n} 
	&:= \frac{\bR_i + \sum_{j \in \N_i'}\sum_{t=1}^{n} Z_{j,t}\Delta_{r,j}(t)}{\sum_{j \in \N_i'}X_{j,n}}  \cr
	&= \frac{\sum_{j \in \N_i^{'}} U_{j,n}X_{j,n}}{\sum_{j \in \N_i^{'}} X_{j,n}}.
\end{align}
Note that $ S_{i,0} = \frac{\bR_i}{\bT_i} $. $ S_{i,n} $ is in fact a function of the random draw variables of the network, and in particular of $ \{Z_j^{n}\}_{j \in \N_i'} $, but for ease of notation, when the arguments are clear, we write $ S_{i,n}\bah{(Z_1^{n},\cdots,Z_N^{n})} = S_{i,n} $. Then the conditional probability of drawing a red ball from the super urn of node $i$ at time $n$
given the complete network history, i.e. given all the past $n-1$ draw variables for each node in the network $\{Z_j^{n-1}\}_{j=1}^N=\{(Z_{1,1},\cdots,Z_{1,n-1}),\cdots,(Z_{N,1},\cdots,Z_{N,n-1})\}$, satisfies  
\begin{align}\label{eq:cond_eq}
	P&\left(Z_{i,n} = 1 | \{Z_j^{n-1}\}_{j=1}^N\right) \cr
	&= \frac{\bR_i + \sum_{j \in \N_i'}\sum_{t=1}^{n-1} Z_{j,t}\Delta_{r,j}(t)}{\sum_{j \in \N_i'}X_{j,n-1}} \nonumber \\
	&= S_{i,n-1}.
\end{align}
That is, the conditional probability of drawing a red ball for node $ i $ at time $ n $ given the entire past $ \{Z_j^{n-1}\}_{j=1}^N $ is the proportion of red balls in its super urn, $ S_{i,n-1} $. This is \bah{however} analogous to the original Polya case, but instead of relying on the individual proportion of red balls $ U_n $ to describe the conditional probability of drawing red balls, we use the super urn proportion of red balls since we now draw from there.

\begin{remark}\longthmtitle{Non-Markovity}
\bah{While~\eqref{eq:cond_eq} may appear to suggest some sort of Markovity property, the process is non-Markovian in general. This can easily be seen due to the fact that a draw at time $ n $ requires knowledge of all previous draws for the entire neighbourhood.}
\end{remark}

A main objective throughout the rest of this paper is to study the evolution and stochastic properties of the process defined above. Using the above conditional probability, we can determine the $n$-fold joint probability of the entire network $\G$: for $ a_i^{n} \in \{0,1\}^n$, $i \in \{1,...,N\}$, we have that
\begin{align}
	&P_{\G}^{(n)}(a_1^n,\cdots,a_N^n) \cr
	&:= P\left(\{Z_i^n = a_i^n\}_{i=1}^N\right) \cr
	&= \prod_{t=1}^nP\left(\{Z_{i,t}=a_{i,t}\}_{i=1}^N\ | \ \{Z_i^{t-1} = a_i^{t-1}\}_{i=1}^N\right) \cr
	&= \prod_{t=1}^n\prod_{i=1}^N \Big(S_{i,t-1}\Big)^{a_{i,t}}\Big(1 - S_{i,t-1}\Big)^{1-a_{i,t}},
	\label{eq:joint_eq}
\end{align}
where $S_{i,t} = S_{i,t}(a_1^t,\cdots,a_N^t) $ is defined in \eqref{eq:cond_eq}. \bah{With the above explicit joint distribution, it is possible to determine the distributions of each node's process. More specifically, using~\eqref{eq:joint_eq}, the $n$-fold distribution of node $i$'s process at time $t \geq n$ is}
\begin{align*}
P_{i,t}^{(n)}&(a_{i,t-n+1},...,a_{i,t}) := \sum_{\substack{a_i^{t-n} \in \{0,1\}^{t-n} \\ a_j^t \in \{0,1\}^t : j\ne i}} P_{\G}^{(t)}(a_1^t,\cdots,a_N^t).
\end{align*} 
\bah{In order to measure the spread of contagion in the network at any given time, we wish to see how likely it is, on average, for a node to be infected at that instant.} We \bah{hence} define the {\em average infection \bah{rate}} in the network at time $n$ as \bah{the average marginal probability of drawing a red ball,}
\[
\tilde{I}_n := \frac{1}{N} \Nsum P(Z_{i,n}=1) = \frac{1}{N} \Nsum P_{i,n}^{(1)}(1).
\]
Note that $ \tilde{I}_n $ is a function of the network topology $ (V,\E) $, the initial placement of balls $ R_i $ and $ B_i $, the draw processes $ \{Z_{i,t}\}_{t=1}^n $, and number of balls added $ \{\Delta_{r,i}(t)\}_{t=1}^n $ and $ \{\Delta_{b,i}(t)\}_{t=1}^n $ for each node $ i \in V $. Unfortunately for an arbitrary network, the above quantity does not yield an exact analytical formula (except in the simple case of complete networks). As such, in general it is hard to mathematically analyze the asymptotic behaviour of $\tilde{I}_n$, which we wish to minimize when attempting to cure an epidemic. Instead we examine the asymptotic stochastic behavior of a closely related variable given by the average individual proportion of red balls at time $n$, namely
\[
	\tilde{U}_n  := \frac{1}{N} \Nsum U_{i,n}, 
\]
which we call the {\em network susceptibility}. \bah{This quantity is related to the conditional probability of drawing a red ball, as seen in~\eqref{eq:S_n}. Since the individual urn of node $ i $ is in every super urn in the neighbourhood, if $ U_{i,n} $ increases then $ S_{j,n} $ increases for every $ j \in \N_i' $, and hence given the past history those nodes are more likely to exhibit infected behaviour as seen from~\eqref{eq:cond_eq}.} Note that similarly to $ \tilde{I}_n $, $ \tilde{U}_n $ is a function of the network variables.

\begin{remark}\longthmtitle{Finite Memory}\label{remark:finite_memory}
It is worth pointing out that a practical adaptation to our model can be considered, where urns have ``finite memory'' in the sense that the balls added after each draw are only kept in each node's urn for a finite number of future draws. This model is developed in~\cite{FA-TF:94} for the classical Polya process in the context of modelling communication channels, where it is shown that the resulting finite memory contagion process is stationary, Markovian and ergodic. \bah{We present the following result that states that in this case the entire state is Markovian and hence it is a limited reinforcement model, but leave an in-depth investigation to a future work.} \oprocend
\end{remark}
\begin{proposition}\longthmtitle{Finite Memory Markovity}
\bah{The entire state of the network Polya contagion process $ \Zprocess{n} $ with finite memory $ M $ is Markovian with memory $ M $.}
\end{proposition}
\begin{proof}
By~\eqref{eq:X_n} and~\eqref{eq:cond_eq} and the fact that added balls are removed after $ M $ steps, we have for $ n > M $ that
\begin{align*}
	&P\left(Z_{i,n} = 1 \ | \ \{Z_j^{n-1}\}_{j=1}^N \right) \\
	%&= \frac{\bR_i + \sum_{j \in \N_i'}\left(\sum_{t=1}^{n-1} Z_{j,t}\Delta_{r,j}(t) - \sum_{t=1}^{n - M - 1}Z_{j,t}\Delta_{r,j}(t)\right)}{\sum_{j \in \N_i'}X_{j,n-1} - X_{j,n-M-1}} \\
	%&= \frac{\bR_i + \sum_{j \in \N_i'}\left(\sum_{t=1}^{n-1} Z_{j,t}\Delta_{r,j}(t) - \sum_{t=1}^{n - M - 1}Z_{j,t}\Delta_{r,j}(t)\right)}{\sum_{j \in \N_i'}T_j + \sum_{t=1}^{n-1} Z_{j,t}\Delta_{r,j}(t) + (1-Z_{j,t})\Delta_{b,j}(t) - \sum_{t=1}^{n-M-1} Z_{j,t}\Delta_{r,j}(t) + (1-Z_{j,t})\Delta_{b,j}(t)} \\
	&= \frac{\bR_i + \sum_{j \in \N_i'}\sum_{t=n-M}^{n-1} Z_{j,t}\Delta_{r,j}(t) }{\bT_i + \sum_{j \in \N_i'}\sum_{t=n-M}^{n-1} Z_{j,t}\Delta_{r,j}(t) + (1-Z_{j,t})\Delta_{b,j}(t)} \\
	&= P\left(Z_{i,n} = 1 \ | \ \{Z_{j,n-M}^{n-1}\}_{j=1}^N \right).
\end{align*}

Using the above result along with conditional independence, for $ (a_1,\ldots,a_N) \in \{0,1\}^N $ we have for $ n > M $ that
\begin{align*}
	&P\left(Z_{1,n} = a_1,\ldots,Z_{N,n} = a_N \ | \ \{Z_j^{n-1}\}_{j=1}^N \right) \\
	&= \prod_{i=1}^N P\left(Z_{i,n} = 1 \ | \ \{Z_j^{n-1}\}_{j=1}^N \right) \\
	&= \prod_{i=1}^N P\left(Z_{i,n} = 1 \ | \ \{Z_{j,n-M}^{n-1}\}_{j=1}^N \right) \\
	&= P\left(Z_{1,n} = a_1,\ldots,Z_{N,n} = a_N \ | \ \{Z_{j,n-M}^{n-1}\}_{j=1}^N \right),
\end{align*}
and hence the entire network process $ \Zprocess{n} $ is Markovian with memory $ M $.
\end{proof}

\section{Stochastic Properties}\label{sec:stochastic_properties}

We next examine the stochastic properties of the network contagion process. We assume throughout the beginning of this section that $\Delta_{r,i}(t)=\Delta_{b,i}(t)=\Delta>0$, for all $ i\in V$ and times $t$; that is the net number of red and black balls added are equal and constant in time for all nodes. In the case of a complete network, the composition of every nodes' super urn is identical, since there is only one super urn that is being drawn from. Thus for a complete network the super urn model is analogous to one urn where multiple draws occur with replacement, which has been recently studied in detail~\cite{MC-MK:13}. However, the analysis in~\cite{MC-MK:13} is carried out in an aggregate sense, i.e., only for the entire urn and not individual processes. Unfortunately, this aggregate approach does not work in a network setting, whereas the super urn model proposed here is applicable.

\subsection{Complete Network Marginal Distributions}\label{subsec:comp_dist}

We first focus on the special case of complete networks to derive some useful probability distributions; later on, we will obtain other stochastic properties that apply to more general networks.

Given that the network is complete, we focus on one of the nodes, say $ i \in V$. For ease of notation, we define $ \bT_j = \sum_{k=1}^N T_k =: \bT $, and similarly, $\bR_j =: \bR, \bB_j =: \bB $, for all $ j \in V $. Defining the events $A_{n-1} = \{Z_{i,n-1}=a_{n-1},...,Z_{i,1} = a_1\}$ and $ W_{n-1} = \{A_{n-1}, \{Z_{j,1}^{n-1} = b_{j}^{n-1}\}_{j\neq i}\}$ with $b_j^{n-1} \in \{0,1\}^{n-1}$, and parameters $ \rho = \frac{\bR}{\bT} $ and $ \delta = \frac{N\Delta}{\bT} $, we can write using~\eqref{eq:cond_eq} under the above assumptions, that 
\begin{align}\label{eq:A_n}
	&P(Z_{i,n}=1,A_{n-1}) \cr
	=& \sum_{\mathclap{\substack{b_{j}^{n-1}\in\{0,1\}^{n-1} : j\neq i}}} P(Z_{i,n} = 1 | W_{n-1}) P(W_{n-1}) \cr
	%=& \sum_{b_{j}^{n-1}} \frac{\left[\bR + \sum_{t=1}^{n-1}\left(\Delta_{r,i}(t)a_t + \sum_{j\neq i} \Delta_{r.j}(t)b_{j,t}\right)\right]P(W_{n-1})}{\bT + \sum_{t=1}^{n-1} a_{t}\Delta_{r,i}(t) + (1-a_{t})\Delta_{b,i}(t)+\sum_{j\neq i}b_{j,t}\Delta_{r,j}(t) + (1-b_{j,t})\Delta_{b,j}(t)} \cr
	%\qquad\text{by~\eqref{eq:cond_eq}} \cr	
	=& \sum_{\mathclap{b_{j}^{n-1}:j\neq i}} \frac{\bR + \Delta\sum_{t=1}^{n-1}\left(a_t + \sum_{j\neq i} b_{j,t}\right)}{\bT + \sum_{t=1}^{n-1} \Delta +\sum_{j\neq i}\Delta} P(W_{n-1}) \cr
	%\qquad\text{since $\Delta$ is constant} \cr
	=& \sum_{\mathclap{b_{j}^{n-1}:j\neq i}} \frac{\rho + \frac{\delta}{N}\sum_{t=1}^{n-1}\left(a_t + \sum_{j\neq i} b_{j,t}\right)}{1 + (n-1)\delta}P(W_{n-1}) \cr
	=& \sum_{\mathclap{b_{j}^{n-1}:j\neq i}}\enskip\enskip\Bigg[\rho \frac{P(A_{n-1},\{Z_{j}^{n-1} = b_{j}^{n-1}\}_{j\neq i})}{1+(n-1)\delta} \cr
	&+ \frac{\delta}{N}\sum_{t=1}^{n-1}\Bigg(a_t \frac{P(A_{n-1},\{Z_{j}^{n-1} = b_{j}^{n-1}\}_{j\neq i})}{1+(n-1)\delta}\cr
	&+ \sum_{j\neq i} \frac{b_{j,t}P(A_{n-1},\{Z_{j}^{n-1} = b_{j}^{n-1}\}_{j\neq i})}{1+(n-1)\delta} \Bigg)\Bigg].
\end{align}
%\margin{Mikhail, in the last equality above, the sum $ \sum_{\mathclap{b_{j}^{n-1},j\neq i}} $ is over all three terms, correct? This is not clear the way it is written, I think. Am I missing something?}
By examining an arbitrary term $k\neq i$ in the final sum above, for fixed $t\in\{1,...,n-1\}$, we can sum out all the other draw variables:
\begin{align}\label{eq:A_n_marg_1}
	&\sum_{\substack{b_{j}^{n-1}\in\{0,1\}^{n-1} : j\neq i}} b_{k,t}P(A_{n-1},\{Z_{j}^{n-1} = b_{j}^{n-1}\}_{j\neq i}) \cr
	&= \sum_{b_k^{n-1}\in\{0,1\}^{n-1}} b_{k,t}P(A_{n-1},Z_{k}^{n-1} = b_{k}^{n-1}) \cr
	%&&\text{by summing out the other nodes' draws}\cr
	&= \sum_{b_{k,t}\in\{0,1\}} b_{k,t}P(A_{n-1},Z_{k,t} = b_{k,t}) \nonumber \\
	%&&\text{by summing out all times except $ t $}\cr
	&= P(A_{n-1},Z_{k,t}=1).
\end{align}
Further, by the law of total probability,
\begin{align}\label{eq:A_n_marg_2}
 &\sum_{\mathclap{\substack{b_{j}^{n-1}\in\{0,1\}^{n-1} : j\neq i}}}P(A_{n-1},\{Z_{j}^{n-1} = b_{j}^{n-1}\}_{j\neq i}) = P(A_{n-1}).
\end{align}
So using~\eqref{eq:A_n_marg_1} and~\eqref{eq:A_n_marg_2},~\eqref{eq:A_n} becomes
\begin{align*}
	& \frac{ \rho P(A_{n-1}) + \frac{\delta}{N}\quad\,\mathclap{\sum_{t=1}^{n-1}}\quad\left[a_tP(A_{n-1}) + \enskip\enskip\mathclap{\sum_{j\neq i}}\quad\, P(A_{n-1},Z_{j,t}=1)\right]}{1+(n-1)\delta}
\end{align*}
Thus, using the law of total probability, we have
\begin{align}\label{eq:marg}
	P(Z_{i,n} = 1) &= \quad\sum_{\mathclap{a^{n-1}\in\{0,1\}^{n-1}}} P(Z_{i,n}=1,A_{n-1}) \cr
	&= \sum_{a^{n-1}} \frac{ \rho P(A_{n-1}) + \frac{\delta}{N}\sum_{t=1}^{n-1} a_tP(A_{n-1})}{1+(n-1)\delta} \cr
	&\qquad\qquad+ \frac{\frac{\delta}{N}\sum_{t=1}^{n-1}\sum_{j\neq i}P(A_{n-1},Z_{j,t}=1)}{1 + (n-1)\delta} \cr
		&= \frac{\rho + \frac{\delta}{N} \sum_{t=1}^{n-1}\sum_{j=1}^N P(Z_{j,t}=1)}{1 + (n-1)\delta}.
\end{align}

An interesting corollary of this derivation is as follows.

\begin{lemma}\longthmtitle{Complete Network Marginal Distribution}\label{lem:approx_1_diml}
The $1$-dimensional marginal distribution of node $i$'s contagion draw process $\{Z_{i,n}\}_{n=1}^{\infty}$ for the $N$-node complete network is given by
\[
P_{i,n}^{(1)} = P(Z_{i,n}  = a) = \rho^a(1-\rho)^{1-a},
\]
where $i\in V$, $n\geq 1$, and $a \in\{0,1\}$. 
\end{lemma}
\begin{proof}
We proceed using strong induction on $ n \geq 1 $, showing that $P(Z_{i,n}=1)=\rho$, for all nodes $ i \in V $ and all $n$. The base case readily holds, since at time $ n = 1 $,
\begin{align*}
	P(Z_{1,1} = 1) &= \cdots = P(Z_{N,1} = 1) = \frac{\sum_{i=1}^N R_i}{\sum_{i=1}^N T_i} = \rho.
\end{align*}
Now, assuming that $P(Z_{j,t} = 1) = \rho $ for all $ j \in V $ and $ t\leq n$ and using~\eqref{eq:marg}, we have
\begin{align*}
	P(Z_{i,n+1} = 1) 
		&= \frac{\rho + \frac{\delta}{N} \sum_{t=1}^{n}\sum_{j=1}^N P(Z_{j,t}=1)}{1 + n\delta} \\
		&= \frac{\rho + \sum_{t=1}^n\frac{\delta}{N}N\rho}{1 + n\delta} \\
		&= \frac{\rho + \delta\tsum\rho}{1+n\delta}	=\rho,
\end{align*}
which completes the induction argument. The result now follows using the fact that
\begin{align*}
	&P(Z_{j,n}=1) + P(Z_{j,n}=0) = 1 \Rightarrow P(Z_{j,n}=0) = 1-\rho,
\end{align*}
for all $ j \in V $ and all $n$.
\end{proof}

We next show that each node's draw process is not stationary in general, and hence is different from the classical $\Polya(\rho_c,\delta_c)$ process.
\begin{remark}\longthmtitle{Non-Stationarity of the Network Contagion Process}\label{rem:not_stat}
Consider a $2$-node complete network. Then, using~\eqref{eq:joint_eq}, one can obtain (after some simplifications) that 
\begin{align*}
	&P(Z_{1,2} = 1, Z_{1,1} = 1) = \rho\frac{\rho+(1+\rho)\frac{\delta}{2}}{1+\delta}, \\
	&P(Z_{1,3} = 1, Z_{1,2} = 1) \\
	&= \sum_{\mathclap{\substack{a_1\in\{0,1\} \\ b^3\in\{0,1\}^3}}}P(Z_{1,1} = a_1,\{Z_{1,t} = 1\}_{t=2}^3, \{Z_{2,t} = b_t\}_{t=1}^3) \\
	&= \rho\frac{4\rho + \delta(2+ 14\rho) + \delta^2(6+14\rho) + \delta^3(5+3\rho)}{4(1 + \delta)^2(1+2\delta)},
\end{align*}
%\margin{There is something off here: $ a_2 $, $a_3 $?}
and hence the network process is not stationary.\oprocend
\end{remark}

Since every exchangeable process is necessarily stationary, Remark~\ref{rem:not_stat} implies that the network Polya process is not exchangeable in general. However, some notions of stationarity remain; in our next result, we will see that there is a consistent relationship between the draws at the $1$st and $n$th time steps.

\begin{lemma}\longthmtitle{Complete Network \boldmath{$(n,1)$}-step Marginal Probability}\label{lem:approx_n1_diml}
For the complete network, the $2$-dimensional marginal probability that node $i$'s draw variables at times $n$ and $1$ are both one is given by
%\margin{I don't think this notation is consistent with the rest, and also I am not sure we need it. I would remove this and just say $ P(Z_{i,n} = 1, Z_{i,1} = 1)=.... $}
\[
P(Z_{i,n} = 1, Z_{i,1} = 1) = \rho\frac{\rho + (1 + (N-1)\rho)\frac{\delta}{N}}{1 + \delta},
\]
for $i\in V$, $n\geq 2$. Furthermore, for any other node $k$, 
\[
	P(Z_{k,n} = 1, Z_{i,1} = 1) = \rho\frac{\rho + (1 + (N-1)\rho)\frac{\delta}{N}}{1 + \delta}.
\]
\end{lemma}
\begin{proof}
By Lemma~\ref{lem:approx_1_diml} we have that $P(Z_{k,1}=1)=\rho$ for all $ k \in V$, so it is enough to show that
\begin{equation}\label{eq:aux4.3}
P(Z_{k,n}=1\ | \ Z_{i,1}=1)=\frac{\rho+(1+(N-1)\rho)\frac{\delta}{N}}{1+\delta}
\end{equation}
for all $ n $ and nodes $i$ and $k$. Using the law of total probability,~\eqref{eq:cond_eq}, and after some simplifications, with defining $W_{n-1}=\{Z_{i,2}^{n-1}=a_2^{n-1},\{Z_{j,1}^{n-1}=b_{j,1}^{n-1}\}_{j\neq i} \}$ \bah{(note that $ W_{n-1} $ is a function of $ \{b_{j,1}^{n-1}\}_{j\neq i} $, but for simplicity we omit this)}, we have that
%\margin{I changed two of the substacks here}
\begin{align*}
	&P(Z_{i,n}=1\ | \ Z_{i,1}=1) \cr
	=&\quad\sum_{\mathclap{\substack{a_2^{n-1} \in \{0,1\}^{n-2} \\ b_{j,1}^{n-1}\in\{0,1\}^{n-1} : j\neq i}}} P(Z_{i,n}=1\ | \ Z_{i,1}=1,W_{n-1}) P(W_{n-1}\  | \ Z_{i,1}=1) \cr
	=&  \enskip\sum_{\mathclap{\substack{a_2^{n-1}, b_{j,1}^{n-1}:j\neq i}}} \frac{\rho+\frac{\delta}{N}(1+\sum_{t=2}^{n-1}a_t + \sum_{t=1}^{n-1}\sum_{j\neq i}b_{j,t})}{1+(n-1)\delta} \cr
	&\qquad\times P(W_{n-1}\  | \ Z_{i,1}=1) \cr
	%&&\text{by~\eqref{eq:cond_eq}} \cr
	%=& \frac{1}{1+(n-1)\delta}\sum_{\substack{a_2^{n-1} \in \{0,1\}^{n-2} \\ b_{j,1}^{n-1}\in\{0,1\}^{n-1} \\ j\neq i}} \Bigg[\rho + \frac{\delta}{N} + \frac{\delta}{N}\sum_{t=2}^{n-1}a_t + \frac{\delta}{N}\sum_{t=1}^{n-1}\sum_{j\neq i}b_{j,t}\Bigg]P(W\  | \ Z_{i,1}=1) \cr
	=& \enskip\sum_{\mathclap{\substack{a_2^{n-1}, b_{j,1}^{n-1}:j\neq i}}}\qquad \Bigg[\left(\rho + \frac{\delta}{N}\right)\frac{P(W_{n-1}\  | \ Z_{i,1}=1)}{1+(n-1)\delta} \cr
	&\quad\qquad\qquad+ \frac{\delta}{N}\sum_{t=2}^{n-1} \frac{a_tP(W_{n-1}\  | \ Z_{i,1}=1)}{1+(n-1)\delta} \cr
	&\quad\qquad\qquad+ \frac{\delta}{N}\sum_{t=1}^{n-1} \frac{b_{j,t}P(W_{n-1}\  | \ Z_{i,1}=1)}{1+(n-1)\delta}\Bigg].
\end{align*}
Then, after arranging terms and using the law of total probability for
\begin{align*}
\qquad\sum_{\mathclap{\substack{a_2^{n-1}, b_{j,1}^{n-1}: j\neq i}}} b_{j,t}P(W_{n-1} \ | \ Z_{i,1}) = P(Z_{j,t}=1 \ | \ Z_{i,1} = 1),
\end{align*}
we have
%\margin{I would say: we have that $ P(Z_{i,n}=1\ | \ Z_{i,1}=1) = $. Also, I think there could more details about the first line of equality.}
\begin{align}\label{aux:4.3-2}
	&P(Z_{i,n}=1\ | \ Z_{i,1}=1) \cr
	&= \frac{(\rho + \frac{\delta}{N})(1)}{1 + (n-1)\delta} + \frac{\frac{\delta}{N}\sum_{t=2}^{n-1}P(Z_{i,t}=1|Z_{i,1}=1)}{1+(n-1)\delta} \cr
	&\quad+ \frac{\frac{\delta}{N}\sum_{t=1}^{n-1}\sum_{j\neq i}P(Z_{j,t}=1|Z_{i,1}=1)}{1+(n-1)\delta} \cr
	&= \frac{\rho + \frac{\delta}{N}\sum_{j\neq i}P(Z_{j,1}=1)}{1+(n-1)\delta} \cr
	&\quad+ \frac{\frac{\delta}{N}\left[1 + \sum_{j=1}^N\sum_{t=2}^{n-1} P(Z_{j,t}=1|Z_{i,1}=1)\right]}{1+(n-1)\delta} \cr
	%=& \frac{\rho + \frac{\delta}{N}(1 + \sum_{j\neq i}\rho)}{1+(n-1)\delta} \cr
	%&+ \frac{+ \sum_{j=1}^N\sum_{t=2}^{n-1} P(Z_{j,t}=1|Z_{i,1}=1))}{1+(n-1)\delta} \cr
	&= \frac{\rho(1 + (N-1)\frac{\delta}{N})}{1+(n-1)\delta} \cr
	&\quad+\frac{\frac{\delta}{N}\left[1 + \sum_{j=1}^N\sum_{t=2}^{n-1}P(Z_{j,t}=1\ | \ Z_{i,1}=1)\right]}{1+(n-1)\delta}.
\end{align}
It can be similarly shown by symmetry of the complete network that~\eqref{aux:4.3-2} holds for $P(Z_{k,n}=1\ | \ Z_{i,1}=1)$ if $k \neq i$.

In order to show~\eqref{eq:aux4.3}, we proceed using strong induction on $n \geq 2$. For the base case, setting $n=2$ in~\eqref{aux:4.3-2}, we have for any $i,k \in V$,
\begin{align*}
	P(Z_{k,2}=1 |  Z_{i,1}=1) &= \frac{\rho(1 + (N-1)\frac{\delta}{N})+\frac{\delta}{N}}{1+ \delta} \cr
	&= \frac{\rho + (1+(N-1)\rho)\frac{\delta}{N}}{1+\delta},
\end{align*}
as desired. Assume now that $P(Z_{k,t}=1\ | \ Z_{i,1}=1)$ is given by~\eqref{eq:aux4.3}, for $2\leq t\leq n-1$ and any $i,k \in V$. Then by~\eqref{aux:4.3-2},
\begin{align*}
	&P(Z_{k,n}=1\ | \ Z_{i,1}=1) \\
	&= \frac{\rho(1 + (N-1)\frac{\delta}{N})}{1+(n-1)\delta} \cr
	&\quad+\frac{\frac{\delta}{N}\left[1 + \sum_{j=1}^N\sum_{t=2}^{n-1}P(Z_{j,t}=1\ | \ Z_{i,1}=1)\right]}{1+(n-1)\delta} \\
	&= \frac{\rho(1 + (N-1)\frac{\delta}{N}) + \frac{\delta}{N}\left[N(n-2)\frac{\rho + (1+(N-1)\rho)\frac{\delta}{N}}{1+\delta}\right]}{1+(n-1)\delta} \\
	&= \bah{\frac{1}{1+(n-1)\delta}\Bigg[(1+\delta)\frac{\rho + (1+(N-1)\rho)\frac{\delta}{N}}{1+\delta} }\\
	&\bah{\qquad+\delta(n-2)\frac{\rho + (1+(N-1)\rho)\frac{\delta}{N}}{1+\delta}\Bigg]} \\
	%&= \frac{\rho + (1+(N-1)\rho)\frac{\delta}{N}}{1+\delta}\times\frac{(1 + \delta) + \delta(n-2)}{1+(n-1)\delta} \\
	&=\bah{\frac{\rho + (1+(N-1)\rho)\frac{\delta}{N}}{1+\delta},}
\end{align*}
which completes the induction argument.
\end{proof}

Although the draw process is not stationary in general, simulated results suggest that it satisfies some asymptotic stationarity properties, in the sense that given sufficient time the process settles \bah{and deviations become very small in magnitude.
 A representative example is shown in Figure~\ref{fig:asymp_stat} for the $ 2 $-dimensional distribution at times $ n $ and $ n-1 $ in the 5-node network shown in Figure~\ref{subfig:5_ba_network}. }

{ 
	\psfrag{0.195}[rr][rr]{\scriptsize{$ 0.195 $}}
	\psfrag{0.2025}[rr][rr]{\scriptsize{$ 0.2025 $}}
	\psfrag{0.215}[rr][rr]{\scriptsize{$ 0.215 $}}
	\psfrag{500}[tc][bc]{\scriptsize{$ 500 $}}
	\psfrag{1000}[tc][bc]{\scriptsize{$ 1,000 $}}
	\psfrag{P(Z1n=1, Z1n-1=1)}[lc][lc]{\scriptsize{Empirical value}}
	\psfrag{Settled value}[lc][lc]{\scriptsize{Settled value}}
\begin{figure}[!ht]
\centering
\includegraphics[width=0.85\linewidth, height=0.15\textheight]{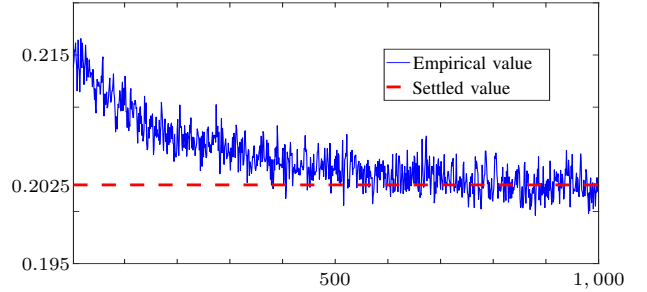}
\caption{Simulated values for $P(Z_{i,n} = 1,Z_{i,n-1}=1)$ for an arbitrary node $ i $ averaged over 50,000 simulated trials on the network shown in Fig.~\ref{subfig:5_ba_network}. All parameters remained constant throughout all trials; see \href{http://bit.ly/2tnBix5}{\color{blue}http://bit.ly/2tnBix5} for a complete list of them.}
\label{fig:asymp_stat}
\end{figure}
}
\vspace{-2.5mm}
\subsection{Martingale Theorems}\label{subsec:mart_thms}
We now turn our attention to the martingale properties of the network contagion process, where we do not assume that the network is necessarily complete. Recall that by \bah{the martingale convergence theorem~\cite{RA-CD:00,GG-DS:01}}, if a process $ \Zprocess{n} $ is a martingale (or supermartingale, or submartingale), there exists a random variable $Z$ such that $ \Zprocess{n} $ converges almost surely to $Z$ as $n \rightarrow \infty$.

\begin{theorem}\longthmtitle{Individual Urn Proportion Martingale}\label{thm:martingale}
For a network $ \G=(V,\E) $, $\Delta_{r,i}(n) = \Delta_{b,i}(n) = \Delta$ and $T_i = T$, for all $i\in V$ and all $n$, the individual proportion of red balls $\{U_{i,n}\}_{n=1}^{\infty}$ is a martingale with respect to the draws for the whole network $ \{Z_n\}_{n=1}^{\infty} = \{(Z_{1,n},...,Z_{N,n})\}_{n=1}^{\infty} $ if and only if, almost surely,
\[
\frac{1}{|\N_i|}\sum_{j \in \N_i}U_{j,n-1} = U_{i,n-1}.
\]
\end{theorem}

\begin{proof}
Using the expression for $ U_{i,n} $,~\eqref{eq:S_n}, and~\eqref{eq:cond_eq}, we have \bah{almost surely}
\begin{align}
  	&E[U_{i,n}\ | \ Z_{n-1} ] \cr
  	&= E\left[ \frac{\Delta Z_{i,n} + U_{i,n-1}(T + (n-1)\Delta)}{T + n\Delta} \ \Big| \ Z_{n-1} \right] \cr
  	&= \frac{U_{i,n-1}(T + (n-1)\Delta)}{T + n\Delta} + \frac{\Delta E[Z_{i,n}\ | \ Z_{n-1}]}{T + n\Delta} \cr
	&= U_{i,n-1}\frac{T+(n-1)\Delta}{T + n\Delta} + \frac{\Delta P(Z_{i,n}=1|Z^{n-1})}{T+n\Delta} \cr
	&= U_{i,n-1}\left(1 - \frac{\Delta}{T + n\Delta}\right) \,\mathclap{+}\enskip \frac{\Delta\sum_{j \in \N_i'} U_{j,n-1}(T+(n-1)\Delta)}{(T+n\Delta)|\N_i'|(T+(n-1)\Delta)} \cr
  	%=& U_{i,n-1}\left(1 - \frac{\Delta}{T + n\Delta}\right) + \Delta\left( \frac{U_{i,n-1} + \sum_{j \in \N_i}U_{j,n-1}}{(T+n\Delta)(|\N_i| + 1)}\right) \cr
	%&&\text{by~\eqref{eq:cond_eq}, since }\Delta_{r,i}(n) = \Delta_{b,i}(n) = \Delta \cr
	&= \bah{U_{i,n-1} - \frac{\Delta U_{i,n-1}}{T+n\Delta} + \Delta\frac{U_{i,n-1} + \sum_{j \in \N_i}U_{j,n-1}}{|\N_i'|(T+n\Delta)} }\cr
  	%&= U_{i,n-1} +\Delta\frac{(1-|\N_i'|)U_{i,n-1} + \sum_{j \in \N_i}U_{j,n-1}}{|\N_i'|(T + n\Delta)} \cr
	&= \bah{U_{i,n-1} +\Delta\frac{\left[\sum_{j \in \N_i}U_{j,n-1}\right] - |\N_i|U_{i,n-1}}{(T + n\Delta)(|\N_i| + 1)}} \cr
  	&= U_{i,n-1} + \frac{\Delta\sum_{j \in \N_i}(U_{j,n-1} - U_{i,n-1})}{(T + n\Delta)(|\N_i| + 1)}.
	\label{eq:individual_mart}
\end{align}
This implies that $\{U_{i,n}\}_{n=1}^{\infty}$ is a martingale with respect to $ \{Z_n\}_{n=1}^{\infty} $ if and only if
\begin{align*}
\sum_{j \in \N_i}U_{j,n-1} - U_{i,n-1} = 0 
\Leftrightarrow \frac{1}{|\N_i|}\sum_{j \in \N_i}U_{j,n-1} = U_{i,n-1}.
\end{align*}
almost surely. 
\end{proof}

If the condition in Theorem~\ref{thm:martingale} holds, \bah{then for any $i$} both $U_{i,n}$ and $\frac{1}{n}\sum_{t=1}^n Z_{i,t}$ converge almost surely to a limit as $n \rightarrow \infty$. However, the condition of Theorem~\ref{thm:martingale}, \bah{barring the trivial single node scenario}, is not verifiable. To resolve this issue, we instead examine the evolution of the average proportion of red balls (i.e., the susceptibility) in a regular network.

\begin{theorem}\longthmtitle{Regular Network Susceptibility Martingale}\label{thm:avg_prop}
For a regular network $ \G=(V,\E) $ with $\Delta_{r,i}(n) = \Delta_{b,i}(n) = \Delta$ and $T_i = T$ for all nodes $i\in V$ and times $n$, the network susceptibility $\{\tilde{U}_n\}_{n=1}^{\infty}$, 
where $\tilde{U}_n = \frac{1}{N}\sum_{i=1}^N U_{i,n}$, is a martingale with respect to $\{Z_n\}_{n=1}^{\infty}$.
\end{theorem}
\begin{proof}
We have, similar to Theorem~\ref{thm:martingale}, that
\begin{align*}
	&E[\tilde{U}_n\ | \ Z_{n-1}] \\
	&= \frac{1}{N}\sum_{i=1}^NE[U_{i,n}\ | \ Z_{n-1}] \\
	&= \frac{1}{N}\sum_{i=1}^N \left[ U_{i,n-1} + \frac{\Delta\sum_{j \in \N_i}\bah{(U_{j,n-1} - U_{i,n-1})}}{(T + n\Delta)(|\N_i| + 1)}\right] \\
	&= \tilde{U}_{n-1} + \sum_{i=1}^N\frac{\Delta\sum_{j \in \N_i}\bah{(U_{j,n-1} - U_{i,n-1})}}{N(T + n\Delta)(|\N_i| + 1)}.
\end{align*}
Let us examine the second term of the last equality. If this term is zero, \bah{$\{\tilde{U}_n\}_{n=1}^{\infty}$ is a martingale with respect to $\{Z_n\}_{n=1}^{\infty}$. We now define the adjacency matrix $[a_{ij}]$ of our network, where the $(i,j)$th entry $a_{ij}$ is 1 if $ (i,j) \in \E $, and 0 otherwise. Since we assumed that our network was undirected, $[a_{ij}]$ is symmetric, i.e., $a_{ij} = a_{ji}$ for all $i,j \in V$. So,}
\begin{align*}
	&\sum_{i=1}^N\frac{\Delta\sum_{j \in \N_i}\bah{(U_{j,n-1} - U_{i,n-1})}}{N(T + n\Delta)(|\N_i| + 1)} \\
	&= \sum_{i=1}^N\frac{\Delta\sum_{j=1}^N a_{ij}(U_{j,n-1} - U_{i,n-1})}{N(T + n\Delta)(|\N_i| + 1)} \\
	&= \frac{\Delta}{N(T + n\Delta)}\sum_{i=1}^N\sum_{j=1}^N\frac{a_{ij}(U_{j,n-1} - U_{i,n-1})}{|\N_i| + 1}
\end{align*}
Now, we examine the sum of the $(i,j)$ and $(j,i)$ components of the double sum, where $(i,j) \in \E$ (otherwise both terms are zero). Recall that $(i,i)\not\in \E$, $\forall i$. We have
\begin{align*}
  	&\frac{a_{ij}(U_{j,n-1} - U_{i,n-1})}{|\N_i| + 1} + \frac{a_{ji}(U_{i,n-1} - U_{j,n-1})}{|\N_j| + 1} \\
	&=\bah{\frac{a_{ij}(U_{j,n-1} - U_{i,n-1})(|\N_j|+1)}{(|\N_i| + 1)(|\N_j| + 1)}} \\
	&\bah{\quad+ \frac{a_{ji}(U_{i,n-1} - U_{j,n-1})(|\N_i|+1)}{(|\N_i| + 1)(|\N_j| + 1)}} \\
  	%&=\frac{a_{ij}\left[U_{j,n-1}(|\N_j| - |\N_i|) + U_{i,n-1}(|\N_i| - |\N_j|)\right]}{(|\N_i| + 1)(|\N_j| + 1)} \\
  	&=\bah{\frac{a_{ij}(|\N_j| - |\N_i|)}{(|\N_i| + 1)(|\N_j| + 1)}\Big(U_{j,n-1} - U_{i,n-1}\Big).}
  \label{eq:avg_prop_zero_term}
\end{align*}
From above, it is clear that this term is zero for all $i$ and $j$ by setting $|\N_j| = |\N_i|$, i.e. in any regular network, and so $\{\tilde{U}_n\}_{n=1}^{\infty}$ is a martingale with respect to $ \{Z_n\}_{n=1}^{\infty} $.
\end{proof}

We next allow the net number of black balls $\Delta_{b,i}(\cdot) $ to evolve stochastically in time as a function of the past draw history in the network in order to steer $ \{U_{i,n}\}_{n=1}^{\infty} $ to a limit for every node $ i $.

\begin{theorem}\longthmtitle{Individual Urn Proportion Categories}\label{thm:U_n_threshold}
In a general network $ \G = (V,\E) $ with $\Delta_{r,i}(n) = \Delta_r $ for all $ n \in \integerspositive $ and $ i \in V $, if we choose $\{\Delta_{b,i}(n)\}_{n=1}^{\infty}$ so that
\[
	\Delta_{b,i}(n) \geq \frac{\Delta_r(1 - U_{i,n-1})S_{i,n-1}}{U_{i,n-1}(1-S_{i,n-1})}
\]
almost surely for all $n \in \integerspositive $ and $ i \in V $ (resp. equal to, less than or equal to) then $ \{U_{i,n}\}_{n=1}^{\infty} $ is a supermartingale (resp. martingale, submartingale) with respect to $ \{Z_n\}_{n=1}^{\infty} $.
\end{theorem}
\begin{proof}
We will start with the case of a supermartingale. That is, we wish to show that almost surely for all $n \in \integerspositive $,
\[
	 E[U_{i,n} \ | \ Z_{n-1}] - U_{i,n-1} \leq 0.
\]
Define $ \bZ{i,n} = \tsum Z_{i,t} $, and take $ X_{i,n} $ as in~\eqref{eq:X_n}. Then, we have almost surely
%	&E[U_{i,n} \ | \ \F_{n-1}] - U_{i,n-1} \\
%	=& \frac{R_i + \Delta_r(\bZ{i,n-1} + E[Z_{i,n} \ | \ \F_{n-1}])}{E[X_{i,n} \ | \ \F_{n-1}]} \\
%	&- \frac{R_i + \Delta_r\bZ{i,n-1}}{X_{i,n-1}} \\
%	=& \frac{\Delta_r E[Z_{i,n} \ | \ \F_{n-1}]}{E[X_{i,n} \ | \ \F_{n-1}]} \\
%	&- \frac{(R_i + \Delta_r\bZ{i,n-1})(E[X_{i,n} | \F_{n-1}] - X_{i,n-1})}{X_{i,n-1}E[X_{i,n} \ | \ \F_{n-1}]} \\
%	=& \frac{\Delta_r E[Z_{i,n} \ | \ \F_{n-1}] - U_{i,n-1}\Delta_{b,i}(n)}{E[X_{i,n} \ | \ \F_{n-1}]} \\
%	&- \frac{U_{i,n-1}E[Z_{i,n} \ | \ \F_{n-1}](\Delta_r - \Delta_{b,i}(n))}{E[X_{i,n} \ | \ \F_{n-1}]}
\begin{align*}
	&U_{i,n} - U_{i,n-1} \\
	&= \frac{R_i + \Delta_r(\bZ{i,n-1} + Z_{i,n})}{X_{i,n}} - \frac{R_i + \Delta_r\bZ{i,n-1}}{X_{i,n-1}} \\
	&= \frac{\Delta_r Z_{i,n}}{X_{i,n} } - \frac{(R_i + \Delta_r\bZ{i,n-1})(X_{i,n} - X_{i,n-1})}{X_{i,n-1}X_{i,n}} \\
	&= \frac{\Delta_r Z_{i,n}}{X_{i,n} } - \frac{U_{i,n-1}(X_{i,n} - X_{i,n-1})}{X_{i,n}}\\
	%&= \frac{\Delta_r Z_{i,n} - U_{i,n-1}(\Delta_r Z_{i,n} +\Delta_{b,i}(n)(1-Z_{i,n}))}{X_{i,n}} \\
	&= \frac{1}{X_{i,n}} \bigg[\Delta_r Z_{i,n} - U_{i,n-1}(\Delta_r Z_{i,n} +\Delta_{b,i}(n)(1-Z_{i,n}))\bigg],
	%=& \frac{\Delta_r Z_{i,n} - U_{i,n-1}\Delta_{b,i}(n)}{X_{i,n} } - \frac{U_{i,n-1}Z_{i,n}(\Delta_r - \Delta_{b,i}(n))}{X_{i,n}},
\end{align*}
since $ X_{i,n} > 0 $ for all $ n \in \integerspositive $ almost surely, we can ignore it. Now, since $ U_{i,n-1} $ is almost surely constant given $Z_{n-1} $,
\[
	E[U_{i,n} \ | \ Z_{n-1}] - U_{i,n-1} \leq 0 \Rightarrow E[U_{i,n} - U_{i,n-1} \ | \ Z_{n-1}] \leq 0.
\]
That is we wish to check if, almost surely,
\[
E\big[\Delta_r Z_{i,n}(1- U_{i,n-1}) - \Delta_{b,i}(n)U_{i,n-1}(1-Z_{i,n}) | Z_{n-1}\big] \leq 0.
\]
Now if
\[
\Delta_{b,i}(n) \geq \frac{\Delta_r(1 - U_{i,n-1})S_{i,n-1}}{U_{i,n-1}(1-S_{i,n-1})}
\]
almost surely, we have
\begin{align*}
	&E\big[\Delta_r Z_{i,n}(1- U_{i,n-1}) - \Delta_{b,i}(n)U_{i,n-1}(1-Z_{i,n}) | Z_{n-1}\big] \\
	&\leq E\Bigg[\Delta_r Z_{i,n}(1- U_{i,n-1}) - U_{i,n-1}(1-Z_{i,n})\\
	&\qquad\qquad\times\frac{\Delta_r(1 - U_{i,n-1})S_{i,n-1}}{U_{i,n-1}(1-S_{i,n-1})} \Bigg| Z_{n-1} \Bigg] \\
	&= \Delta_r(1- U_{i,n-1}) E\left[Z_{i,n} - (1-Z_{i,n})\frac{S_{i,n-1}}{1-S_{i,n-1}} \bigg| Z_{n-1} \right] \\
	%&= \Delta_r(1- U_{i,n-1})\left[E[Z_{i,n}|Z_{n-1}] - (1-E[Z_{i,n}|Z_{n-1}])\frac{S_{i,n-1}}{1-S_{i,n-1}} \right] \\
	&= \Delta_r(1- U_{i,n-1})\left[S_{i,n-1} - (1-S_{i,n-1})\frac{S_{i,n-1}}{1-S_{i,n-1}} \right] \\
	&= 0,
\end{align*}
%%& \frac{\Delta_r E[Z_{i,n} \ | \ \F_{n-1}] - U_{i,n-1}\Delta_{b,i}(n)}{E[X_{i,n} \ | \ \F_{n-1}]} \\
%%	&\leq \frac{U_{i,n-1}E[Z_{i,n} \ | \ \F_{n-1}](\Delta_r - \Delta_{b,i}(n))}{E[X_{i,n} \ | \ \F_{n-1}]} \\
%	& \Delta_{b,i}(n) U_{i,n-1}(1 - E[Z_{i,n} \ | \ \F_{n-1}]) \\
%	&\quad\geq \Delta_r E[Z_{i,n} \ | \ \F_{n-1}](1 - U_{i,n-1}) \\
%	\Rightarrow& \Delta_{b,i}(n) \geq \frac{\Delta_r E[Z_{i,n} \ | \ \F_{n-1}](1 - U_{i,n-1})}{U_{i,n-1}(1 - E[Z_{i,n} \ | \ \F_{n-1}] )} \\
%	\Rightarrow& \Delta_{b,i}(n) \geq \frac{\Delta_r (1 - U_{i,n-1})\sum_{j\in\N_i'}U_{j,n-1}}{U_{i,n-1}(\sum_{j\in\N_i'}X_{j,n-1}(1 - U_{j,n-1}))}
where the second to last equality comes from the fact that $ E[Z_{i,n} | Z_{n-1}] = P(Z_{i,n} = 1 | Z_{n-1}) = S_{i,n-1} $ almost surely by~\eqref{eq:cond_eq}, and that $S_{i,n-1}$ is almost surely constant given $Z_{n-1}$. Thus as long as $ \Delta_{b,i}(n) $ obeys this bound almost surely for all $ n \in \integerspositive $, $ \{U_{i,n}\}_{n=1}^{\infty} $ is a supermartingale with respect to $ \{Z_n\}_{n=1}^{\infty} $. Similarly, if $ \Delta_{b,i}(n) $ is almost surely equal (resp. less than or equal) to this bound, $ \{U_{i,n}\}_{n=1}^{\infty} $ is a martingale (resp. submartingale) with respect to $ \{Z_n\}_{n=1}^{\infty} $.
\end{proof}

Theorem~\ref{thm:U_n_threshold} tells us what bounds for $\{\Delta_{b,i}(n)\}_{n=1}^{\infty}$ must be obeyed almost surely to guarantee that $\{U_{i,n}\}_{n=1}^{\infty} $ admits an asymptotic limit for all $ i \in V $ in any general network. For instance, this tells us that by choosing $ \Delta_{b,i}(t) = 0 $ almost surely for all $ i \in V $ and $ t \in \integerspositive $, $\{U_{i,n}\}_{n=1}^{\infty} $ will be a submartingale and will converge to some limiting random variable. While this result is interesting for modelling contagion, it is especially useful in the context of curing.

\section{Model Approximations}\label{sec:model_approximations}

As previously noted, the dynamics of the network contagion process are complicated, especially when considered on general networks. For this reason, in this section we develop two useful approximations to this process on a general network that allow us to shed some light on its asymptotic behaviour. Throughout this section, unless stated otherwise, we consider general network topologies with $\bah{\Delta_{r,i}(t)} = \Delta_{b,i}(t) = \Delta $ for all $t \in \integerspositive $ and $i \in V$. However, to match the $1$-step and $(n,1)$-step distributions, we make the simplifying assumption that the neighbourhood of each node $ i $ can be represented as a complete network, i.e., all of its neighbours are connected to one another, in order to apply Lemmas~\ref{lem:approx_1_diml} and~\ref{lem:approx_n1_diml}.

\subsection{Approximation: Computational Model}\label{subsec:approx_1}
We now introduce our first approximation technique, where we approximate the contagion process of each node in the network with a classical Polya urn process.
 
\begin{model}{I}\longthmtitle{Computational Model}\label{mod:comp}
We approximate the dynamics of any node $i$'s contagion process using a classical Polya process $\Polya(\rho_c=\rho_i,\delta_c=\dhat_i)$, with
\begin{align*}
\rho_i = \frac{\sum_{j \in \N_i^{'}} R_j}{\sum_{j \in \N_i^{'}} T_j}, \quad\mathrm{and}\quad
 \dhat_i = \argmin_{\tilde{\delta}}\frac{1}{n}D\left(P_{i,n}^{(n)}|| Q_{\rho_i,\tilde{\delta}}^{(n)}\right),
\end{align*}
with
\begin{align*}
&Q_{\rho_i,\tilde{\delta}}^{(n)}(a^n) = \frac{\Gamma\left(\frac{1}{\tilde{\delta}}\right)\Gamma\left(\frac{\rho_i}{\tilde{\delta}} + \bar{a}^n\right)\Gamma\left(\frac{1-\rho_i}{\tilde{\delta}} + n - \bar{a}^n\right)}{\Gamma\left(\frac{1}{\tilde{\delta}} + n\right)\Gamma\left(\frac{\rho_i}{\tilde{\delta}}\right)\Gamma\left(\frac{1-\rho_i}{\tilde{\delta}}\right)},
\end{align*}
where $ \Gamma(\cdot) $ is the Gamma function, $ a^n = (a_1,...,a_n)\in\{0,1\}^n $, and $\bar{a}^n = a_1+\cdots+a_n $. \oprocend
\end{model}

Here $\rho_c$ is chosen to be the proportion of red balls $\rho_i$ in the node's super urn, so that the $1$-dimensional distributions of the classical Polya process and the node process $\{Z_{i,n}\}$ coincide, while $\dhat_i$ is set by performing a minimization to find the value that best fits $Q_{\rho_i,\dhat_i}^{(n)}$ to the distribution of $\Zprocess{i,n}$ of node $i \in V$. We use a divergence measure, denoted by $ D(\cdot || \cdot) $, to observe the quality of the fit.

The explicit derivation of the distribution $Q_{\rho_i,\dhat_i}^{(n)}$ can be found in~\cite{WF:71,NLJ-SK:77}. This method ensures that the fit of $Q_{\rho_i,\dhat_i}^{(n)}$ is as close as possible under the given divergence measure. Since we are measuring the error in using an approximating distribution, we use the Kullback-Leibler divergence~\cite{TC-JT:06}; we thus have that
\begin{align*}
	\dhat_i &= \argmin_{\tilde{\delta}}\frac{1}{n}\sum_{a^n\in\{0,1\}^n}P_{i,n}^{(n)}(a^n)\log{\frac{P_{i,n}^{(n)}(a^n)}{Q_{\rho_i,\tilde{\delta}}^{(n)}a^n)}} \\
		&= \argmax_{\tilde{\delta}}\frac{1}{n}\sum_{a^n\in\{0,1\}^n}P_{i,n}^{(n)}(a^n)\log{Q_{\rho_i,\tilde{\delta}}^{(n)}(a^n)} 
\end{align*}
since $P_{i,n}^{(n)}(a^n)\log{P_{i,n}^{(n)}(a^n)}$ is independent of $\tilde{\delta}$. The approximating process is stationary and exchangeable, as it is a classical Polya process. We also know (from Section~\ref{sec:preliminaries}) that it is non-ergodic with its sample average converging almost surely to the $\Betafun(\frac{\rho_i}{\dhat_i},\frac{1-\rho)i}{\dhat_i})$ distribution. Calculating an analytic expression for the minimizing $\dhat_i$ is not feasible in general, and hence should be performed computationally. However, due to the above minimization, the value of $\dhat_i$ is, by definition, the best way to fit a Polya process to the process $\Zprocess{i,n}$ for a given $n$. 

\subsection{Approximation: Analytical Models}\label{subsec:approx_2}
An alternative to Model~\ref{mod:comp} is to attempt to find approximations whose parameters can be determined analytically. 

\begin{model}{II(a)}\longthmtitle{Large-Network Analytical Model}\label{mod:anal}
For any given node $i$, we approximate the dynamics of its process $\Zprocess{i,n}$ by using a classical Polya process $\Polya(\rho_c=\rho_i,\delta_c=\delp_i)$, with
\begin{align*}
	\rho_i= \frac{\sum_{j \in \N_i^{'}} R_j}{\sum_{j \in \N_i^{'}} T_j}, \quad \mathrm{and} \quad
	\delp_i =\frac{\delta_i}{N+(N-1)\delta_i},
\end{align*}
where $\delta_i=\frac{N \Delta}{\sum_{j\in \N_i^{'}} T_j}$. \oprocend
\end{model}

Here the parameters of the classical Polya process are chosen
by directly matching its first and $(n,1)$-step second-order statistics with those of $\Zprocess{i,n}$. This method avoids the computational burden of the previous model by yielding an analytical expression for the correlation parameter of the classical Polya process.

We next prove that under some stationarity and symmetry assumptions, the contagion process running on each node in the network is statistically identical to the classical Polya process of Model~\ref{mod:anal}.

\begin{lemma}\longthmtitle{Exact Representation}\label{lem:dist_assump}
Suppose that
\begin{enumerate}
	\item\label{lem:it:1} $P(Z_{i,1}=1 \ | \  Z_{j,1}^{n-1}=a^{n-1}) = \rho_i$, and
	\item\label{lem:it:2} $P(Z_{i,t}=1  |   Z_{j,1}^{n-1}=a^{n-1}) = P(Z_{k,n} = 1  | Z_{j,1}^{n-1}=a^{n-1})$,
\end{enumerate}
for all $n\geq 1, 2 \leq t < n$, $i,j,k\in V$, $a^{n-1}\in \{0,1\}^{n-1}$. Then for any node $i$ in a complete network,
$\Zprocess{i,n}$ is given exactly by the $\Polya(\rho_i,\delp_i)$ process.
\end{lemma}
\begin{proof}
For any node $i$, we wish to show that for all $n$, the $n$-dimensional distributions of $\Zprocess{i,n}$ and the $\Polya(\rho_i,\delp_i)$ process are identical. It is enough to show that the conditional probability of one event given the whole past is the same, since any joint probability can be written as a product of conditional probabilities. Let us define the events $ A_{n-1} = \{Z_{i,1}^{n-1}=a^{n-1}\} $ and $ B_{n-1} = \{Z_{j,1}^{n-1}=b_{j,1}^{n-1}\}_{j\neq i}$. Then,
\begin{align*}
&P_{i|n} := P(Z_{i,n}=1\ | \ A_{n-1}) \\
	&= \qquad\sum_{\mathclap{\substack{b_{j,1}^{n-1} \in\{0,1\}^{n-1}:j\neq i}}} P(Z_{i,n}=1 \ | \ A_{n-1} ,B_{n-1}) P(B_{n-1}\ | \ A_{n-1}) \\
	&= \sum_{\mathclap{\substack{b_{j,1}^{n-1}:j\neq i}}} \frac{\rho_i+\frac{\delta_i}{N}\sum_{t=1}^{n-1}(a_t + \sum_{j\neq i}b_{j,t})}{1+(n-1)\delta_i} P(B_{n-1}\ | \ A_{n-1}) \cr
	&= \sum_{\mathclap{\substack{b_{j,1}^{n-1}:j\neq i}}} \frac{\rho_i(1-(N-1)\delp_i)+\delp_i\sum_{t=1}^{n-1}(a_t + \sum_{j\neq i}b_{j,t})}{1+(N(n-2)+1)\delp_i} \\
	&\qquad\times P(B_{n-1} \ | \ A_{n-1}) \cr
	&= \frac{\rho_i(1-(N-1)\delp_i)+\delp_i\sum_{t=1}^{n-1}a_t}{1+(N(n-2)+1)\delp_i}\sum_{\mathclap{\substack{b_{j,1}^{n-1}:j\neq i}}}P(B_{n-1}\ | \ _{n-1}A) \\
	&+ \frac{\delp_i}{1+(N(n-2)+1)\delp_i}\sum_{t=1}^{n-1}\sum_{j\neq i}\enskip\enskip\sum_{\mathclap{\substack{b_{j,1}^{n-1}:j\neq i}}} b_{j,t}P(B_{n-1}\ | \ A_{n-1}) \\
	&= \frac{\left(\rho_i(1-(N-1)\delp_i)+\delp_i\sum_{t=1}^{n-1}a_t\right)\cdot 1}{1+(N(n-2)+1)\delp_i} \\
	&\quad+\frac{\delp_i\sum_{t=1}^{n-1}\sum_{j\neq i}P(Z_{j,t}=1\ | \ A_{n-1})}{1+(N(n-2)+1)\delp_i}.
\end{align*}
Then using assumption~\ref{lem:it:1}, we have
\begin{align*}
	P_{i|n} &= \frac{\rho_i(1-(N-1)\delp_i)+\delp_i\sum_{t=1}^{n-1}a_t +\delp_i(N-1)\rho_i}{1+(N(n-2)+1)\delp_i} \\
	&\quad+ \frac{\delp_i\sum_{t=2}^{n-1}\sum_{j\neq i}P(Z_{j,t}=1\ | \ A_{n-1})}{1+(N(n-2)+1)\delp_i} \\
	&= \frac{\rho_i+\delp_i\sum_{t=1}^{n-1}\left[a_t + \sum_{j\neq i}P(Z_{j,t}=1\ | \ A_{n-1})\right]}{1+(N(n-2)+1)\delp_i}
\end{align*}
Now using assumption~\ref{lem:it:2}, we get
\begin{align*}
P_{i|n}	&= \frac{\rho_i+\delp_i\left(\sum_{t=1}^{n-1}a_t +\sum_{t=2}^{n-1}\sum_{j\neq i}P_{i|n}\right)}{1+(N(n-2)+1)\delp_i} \cr
	&= \frac{\rho+\delp_i\left(\sum_{t=1}^{n-1}a_t +(n-2)(N-1)P_{i|n}\right)}{1+(N(n-2)+1)\delp_i}.
\end{align*}
Thus, we have that
\begin{align*}
&P_{i|n} 	=  \frac{\rho_i+\delp_i\left(\sum_{t=1}^{n-1}a_t +(n-2)(N-1)P_{i|n}\right)}{1+(N(n-2)+1)\delp_i} \\
\Rightarrow& P_{i|n} = \frac{\rho_i+\delp_i\sum_{t=1}^{n-1}a_t}{1+(n-1)\delp_i},
\end{align*}
which is the conditional probability $P(Z_n=1 | Z_1^{n-1}=a^{n-1})$ for a $\Polya(\rho_i,\delp_i)$ process. A similar calculation can be performed for $P(Z_{i,n}=0\ | \ Z_{i,1}^{n-1}=a^{n-1})$.
\end{proof}

Unfortunately in a general network setting assumptions~\ref{lem:it:1} and~\ref{lem:it:2} above do not hold true. However, this result motivates the fact that this analytical approximation is reasonable to use for situations where these \bah{assumptions} hold within tolerable margins of error; empirical evidence indicates that this occurs for large values of $N$, since as $N$ increases the quality of the fit improves. This approximation, nevertheless, drastically reduces the complexity in analyzing the individual contagion draw processes, as closed-form expressions for the process parameters are available.

\begin{model}{II(b)}\longthmtitle{Small-Network Analytic Model}\label{mod:anal_2}
Given any node $i$ in the network with a small to moderate number of nodes, we approximate the dynamics of its contagion process $\Zprocess{n}$ using a $\Polya(\rho_i, \dels_i)$ process, where
\begin{align*}
	\rho_i &= \bah{\frac{\sum_{j \in \N_i^{'}} R_j}{\sum_{j \in \N_i^{'}} T_j},} \quad \mathrm{and} \\
\dels_i &=\frac{\delta_i/N}{N+(N-1)\delta_i/N} = \frac{\delta_i}{N^2+(N-1)\delta_i},
\end{align*}
where $\delta_i=\frac{N \Delta}{\sum_{j\in \N_i^{'}} T_j}$.\oprocend
\end{model}
The idea behind this model is that we want to remove the dependence on the number of nodes $N$ from the parameter $\delta_i = \frac{N\Delta}{\bT_i}$, and so we divide each instance of $\delta_i$ in $\dels_i$ by $N$. \bah{The idea is that as $n$ grows, it eventually becomes significantly larger than the relatively small number of nodes $N$, and so $n|\N_i| \approx n$ for all $ i \in V $. Hence, we may consider that for a sufficiently large time, we have added $n\Delta$ balls to the super urn. Effectively, this means we are using a correlation parameter of $\frac{\Delta}{\bT_i}$ instead of $\delta = \frac{N\Delta}{\bT_i}$. Simulation results confirm that this approximation captures the limit distribution of the original process better than Model~\ref{mod:anal} when the number of nodes is small.} Figure~\ref{fig:hist_comp} displays this relationship. A summary of all models presented in this section, and the scenarios under which they are most suitable, is provided in Table~\ref{tab:approx_usage}.
\begin{table}[ht!]
	\centering
	\caption{Approximation Usage Scenarios}
	\label{tab:approx_usage}
	\begin{tabular}{c | l}
		Model & Usage Scenario \\
		\hline
		~\ref{mod:comp} & Exactness valued over analytic simplicity \\
		~\ref{mod:anal} & Larger values of $N$, i.e., large network \\
		~\ref{mod:anal_2} & Small to moderate values of $N$, i.e., small network
	\end{tabular}
\end{table}
\begin{figure*}
\centering
{
	\psfrag{0.15}[rr][rr]{\scriptsize{}}
	\psfrag{0.2}[cc][cb]{\scriptsize{0.2}}
	\psfrag{0.25}[rr][rr]{\scriptsize{}}
	\psfrag{0.3}[cc][cb]{\scriptsize{0.3}}
	\psfrag{0.35}[rr][rr]{\scriptsize{}}
	\psfrag{0.4}[cc][cb]{\scriptsize{0.4}}
	\psfrag{0.45}[rr][rr]{\scriptsize{}}
	\psfrag{0.5}[cc][cb]{\scriptsize{0.5}}
	\psfrag{0.55}[rr][rr]{\scriptsize{}}
	\psfrag{0.6}[cc][cb]{\scriptsize{0.6}}
	\psfrag{0.65}[rr][rr]{\scriptsize{}}
	\psfrag{0}[rr][rr]{\scriptsize{0}}
	\psfrag{5}[rr][rr]{\scriptsize{5}}
	\psfrag{15}[rr][rr]{\scriptsize{15}}
	\psfrag{Empirical}[lc][lc]{\scriptsize{Empirical}}
	\psfrag{Model I}[lc][lc]{\scriptsize{Model I}}
	\psfrag{Model II(a)}[lc][lc]{\scriptsize{Model II(a)}}
	\psfrag{Model II(b)}[lc][lc]{\scriptsize{Model II(b)}}
\subfigure[10-node complete network histogram.]{ \includegraphics[width=0.4\linewidth, height=0.15\textheight]{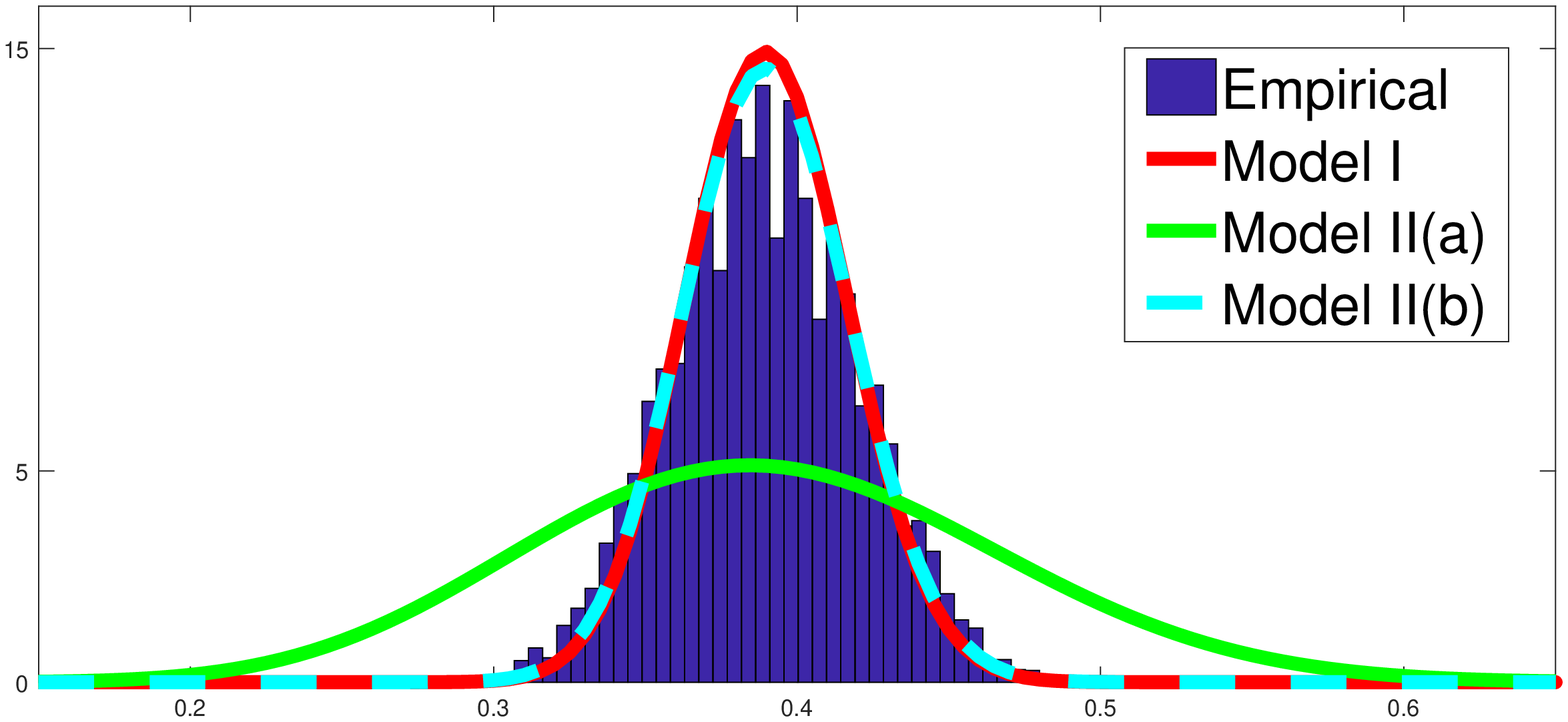} \label{subfig:10_compl}}
}
\qquad
{
	\psfrag{0.45}[cc][cb]{\scriptsize{0.45}}
	\psfrag{0.5}[cc][cb]{\scriptsize{0.5}}
	\psfrag{0.55}[cc][cb]{\scriptsize{0.55}}
	\psfrag{25}[rr][rr]{\scriptsize{25}}
	\psfrag{150}[rr][rr]{\scriptsize{150}}
	\psfrag{Empirical}[lc][lc]{\scriptsize{Empirical}}
	\psfrag{Model I}[lc][lc]{\scriptsize{Model I}}
	\psfrag{Model II(a)}[lc][lc]{\scriptsize{Model II(a)}}
	\psfrag{Model II(b)}[lc][lc]{\scriptsize{Model II(b)}}
\subfigure[100-node complete network histogram.]{ \includegraphics[width=0.4\linewidth, height=0.15\textheight]{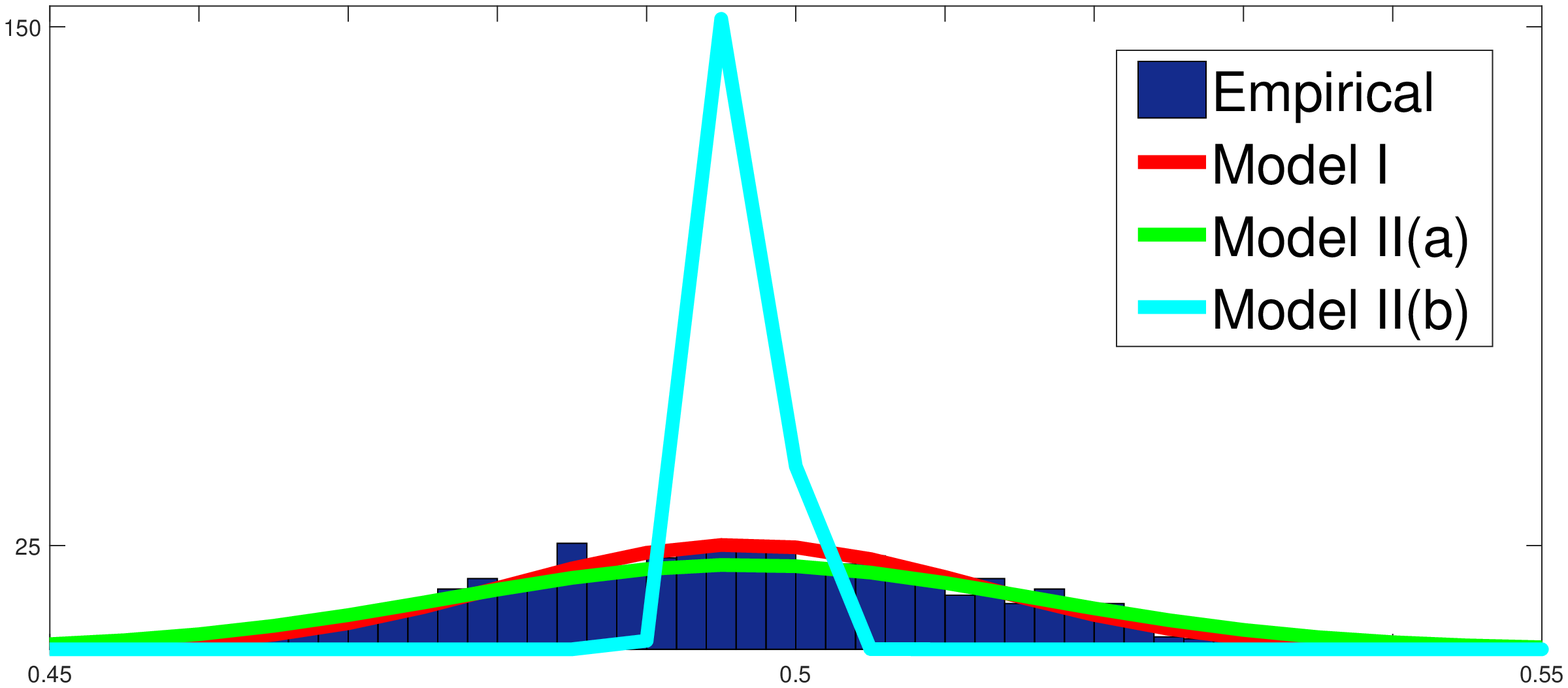} \label{subfig:100_compl}}
}
{
	\psfrag{0}[rr][rr]{\scriptsize{0}}
	\psfrag{0.3}[cc][cb]{\scriptsize{0.3}}
	\psfrag{0.6}[cc][cb]{\scriptsize{0.6}}
	\psfrag{0.9}[cc][cb]{\scriptsize{0.9}}
	\psfrag{2}[rr][rr]{\scriptsize{2}}
	\psfrag{4}[rr][rr]{\scriptsize{4}}
	\psfrag{4.5}[rr][rr]{\scriptsize{4.5}}
	\psfrag{Empirical}[lc][lc]{\scriptsize{Empirical}}
	\psfrag{Model I}[lc][lc]{\scriptsize{Model I}}
	\psfrag{Model II(a)}[lc][lc]{\scriptsize{Model II(a)}}
	\psfrag{Model II(b)}[lc][lc]{\scriptsize{Model II(b)}}
\subfigure[5-node Barabasi-Albert network histogram.]{ \includegraphics[width=0.4\linewidth, height=0.15\textheight]{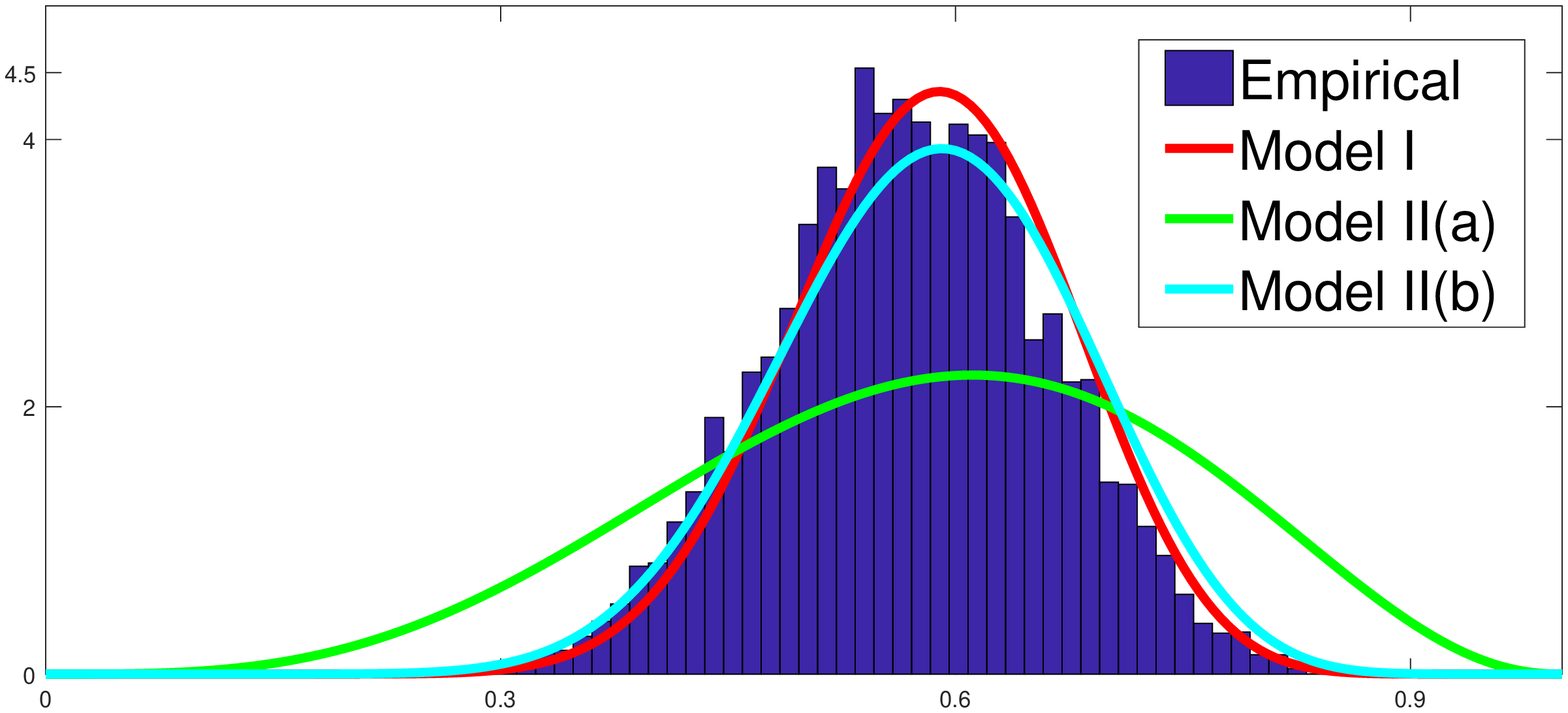} \label{subfig:5_ba}}
}
\qquad
\subfigure[5-node Barabasi-Albert network~\cite{RA-ALB:02}.]{ \includegraphics[width=0.4\linewidth, height=0.15\textheight]{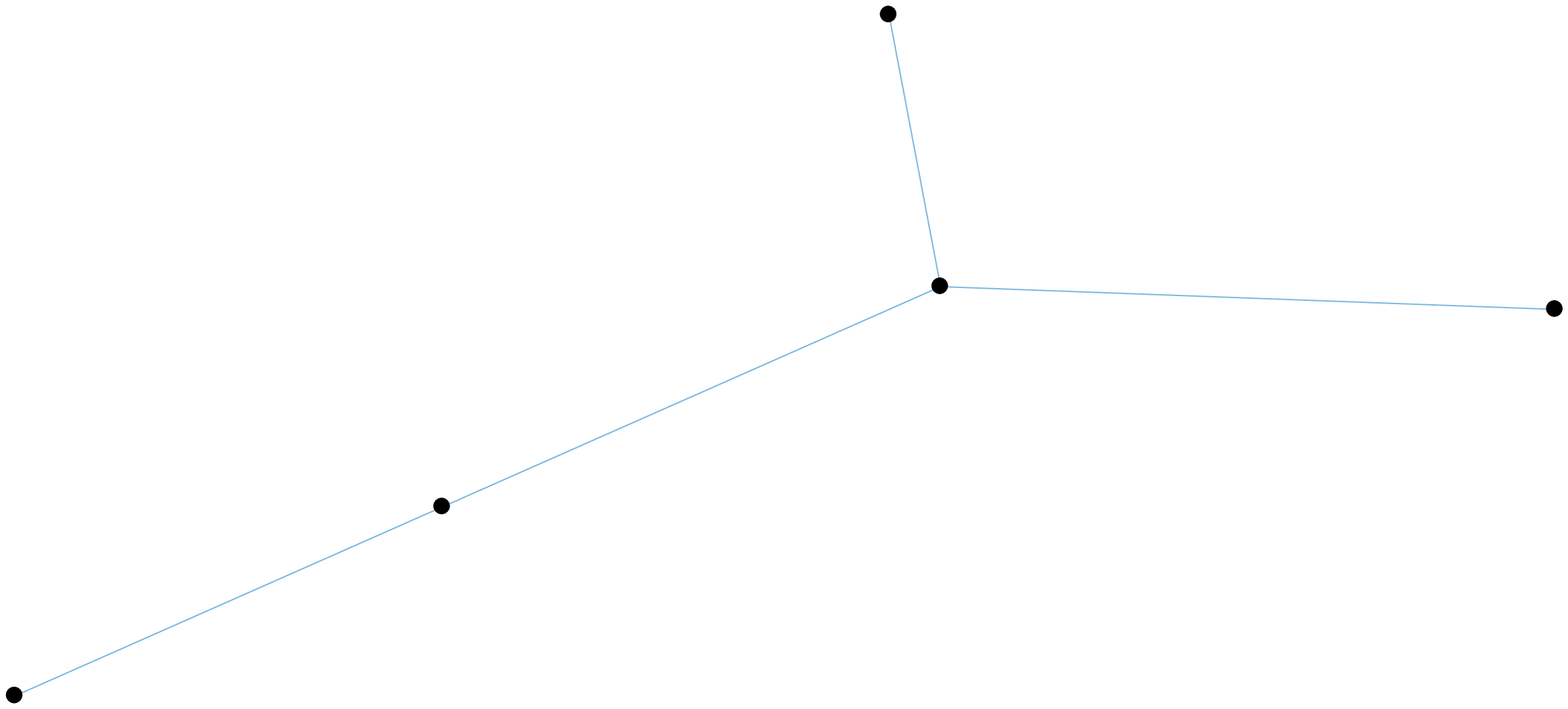} \label{subfig:5_ba_network}} \\
{
	\psfrag{0.3}[cc][cb]{\scriptsize{0.3}}
	\psfrag{0.45}[cc][cb]{\scriptsize{0.45}}
	\psfrag{0.5}[cc][cb]{\scriptsize{0.5}}
	\psfrag{0.55}[cc][cb]{\scriptsize{0.55}}
	\psfrag{0.7}[cc][cb]{\scriptsize{0.7}}
	\psfrag{0}[rr][rr]{\scriptsize{0}}
	\psfrag{5}[rr][rr]{\scriptsize{5}}
	\psfrag{25}[rr][rr]{\scriptsize{25}}
	\psfrag{Empirical}[rr][rr]{\scriptsize{Empirical}}
	\psfrag{Empirical}[lc][lc]{\scriptsize{Empirical}}
	\psfrag{Model I}[lc][lc]{\scriptsize{Model I}}
	\psfrag{Model II(a)}[lc][lc]{\scriptsize{Model II(a)}}
	\psfrag{Model II(b)}[lc][lc]{\scriptsize{Model II(b)}}
\subfigure[100-node Barabasi-Albert network histogram.]{ \includegraphics[width=0.4\linewidth, height=0.15\textheight]{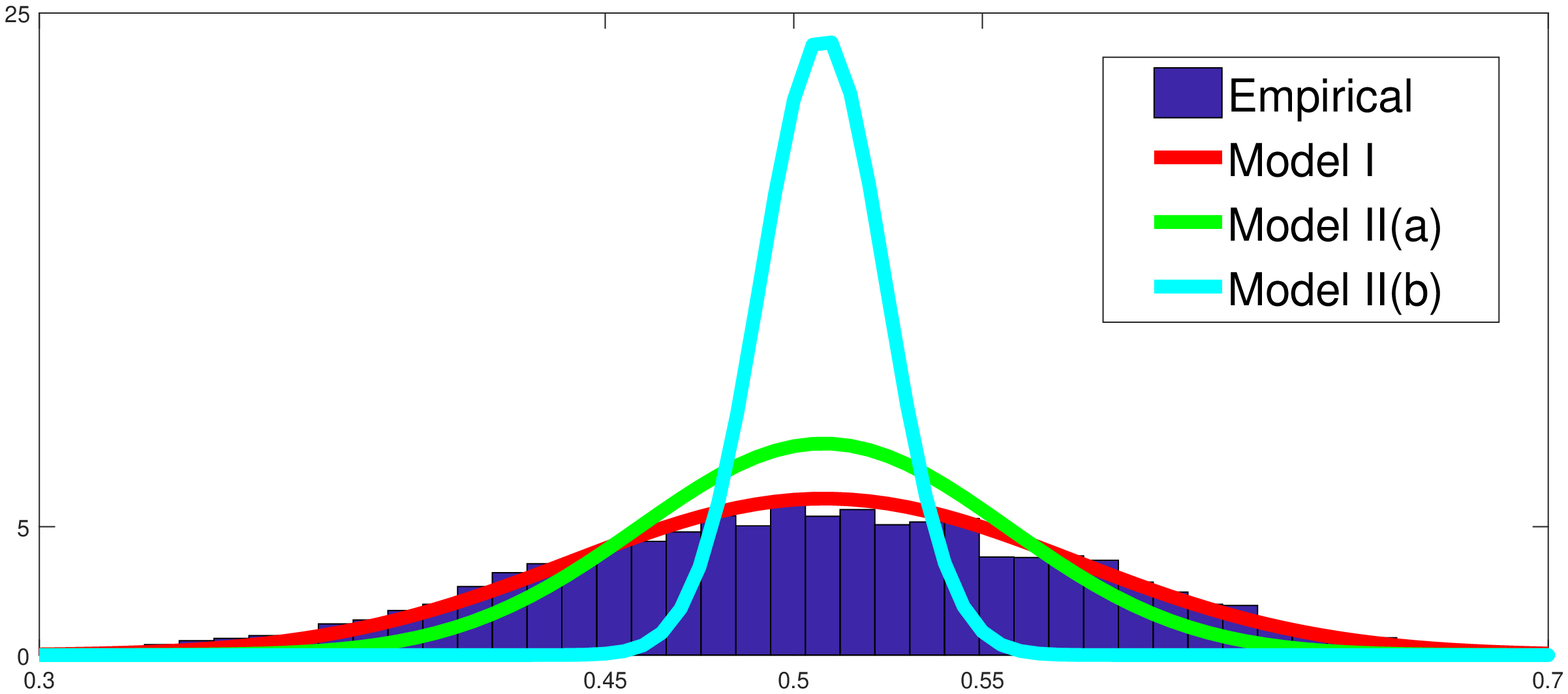} \label{subfig:100_ba}}
}
\qquad
\subfigure[100-node Barabasi-Albert network~\cite{RA-ALB:02}.]{ \includegraphics[width=0.4\linewidth, height=0.15\textheight]{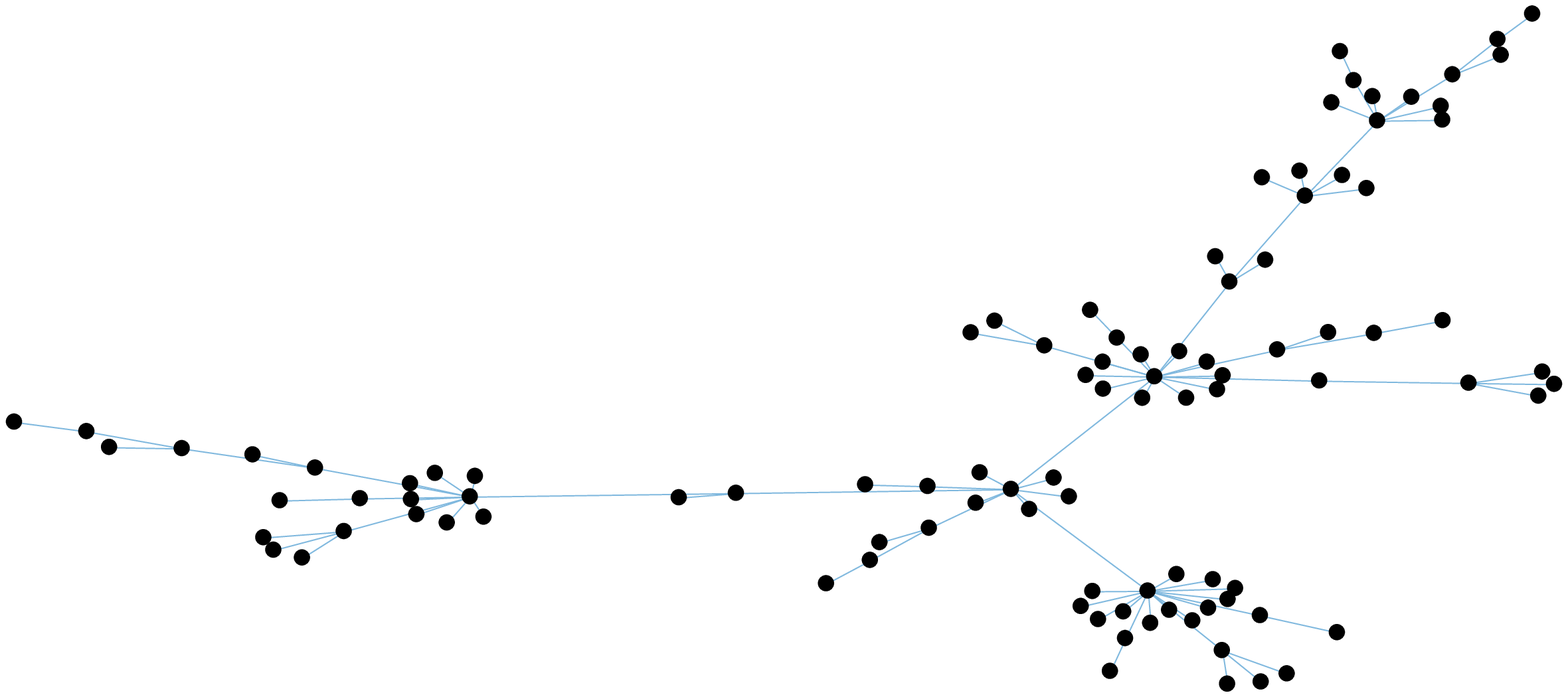} \label{subfig:100_ba_network}}
\caption{Comparison of normalized simulated histograms for the sample average of draws $ \frac{1}{n}\sum_{t=1}^n Z_{i,t} $ and the $\Betafun(\frac{\rho_i}{\dhat_i},\frac{1-\rho_i}{\dhat_i}), \Betafun(\frac{\rho_i}{\delp_i},\frac{1-\rho_i}{\delp_i})$, and $\Betafun(\frac{\rho_i}{\dels_i},\frac{1-\rho_i}{\dels_i}) $ pdfs from Models~\ref{mod:comp},~\ref{mod:anal}, and~\ref{mod:anal_2}, respectively, for an arbitrary node $ i $ with $n = 1000$, averaged over $5,000$ simulated trials. \bah{Here the parameters $ \Delta $ as well as $R_i$ and $B_i$, for all $ i \in V $,  were uniformly randomly assigned for each network, and were consistent throughout all trials. See \href{http://bit.ly/2tnBix5}{\color{blue}http://bit.ly/2tnBix5} for a complete list of all parameters used for each network.}}
\label{fig:hist_comp}
\vspace{-2.5mm}
\end{figure*}

We close this section with numerical demonstrations on the fitness of all models. Figure~\ref{fig:hist_comp} shows a representative comparison between the $\Betafun(\frac{\rho_i}{\delp_i},\frac{1-\rho_i}{\delp_i})$ pdf and the simulated histogram of $\frac{1}{n}\sum_{t=1}^n Z_{i,n}$, where $ n=1000 $, for an arbitrary node $i$ in the given networks. Recall that the $\Betafun(\frac{\rho_i}{\dhat_i},\frac{1-\rho_i}{\dhat_i})$, $\Betafun(\frac{\rho_i}{\delp_i},\frac{1-\rho_i}{\delp_i})$ and $\Betafun(\frac{\rho_i}{\dels_i},\frac{1-\rho_i}{\dels_i})$ pdfs are the distributions of the limit random variables to which the sample average of the draw processes of Models~\ref{mod:comp},~\ref{mod:anal} and~\ref{mod:anal_2} (respectively) converge almost surely, as $ n \rightarrow \infty $ (see Section~\ref{sec:preliminaries}). We use complete networks since they satisfy the assumption that all neighbourhoods are complete, as well as Barabasi-Albert networks which have been shown to be a good model for real-world social networks~\cite{RA-ALB:02} and do not satisfy this assumption; however, our results show that the approximations still fit quite well. As expected, Model~\ref{mod:comp} provides the best approximation in all scenarios, albeit without an analytic expression for its parameters which can provide insight into the behaviour of the underlying process. Model~\ref{mod:anal} fits quite well when the number of nodes in the network is large, as seen in Figures~\ref{subfig:100_compl} and~\ref{subfig:100_ba}, but fits poorly for a small number of nodes, which is evident in Figures~\ref{subfig:10_compl} and~\ref{subfig:5_ba}. Model~\ref{mod:anal_2} is the complement of Model~\ref{mod:anal} in the sense that it fits very well for a small number of nodes but poorly for a large network. Hence if analytic expressions for parameters are desired, Models~\ref{mod:anal} and~\ref{mod:anal_2} can be used depending on the number of nodes to provide approximations that are marginally worse than the computational exactness of Model~\ref{mod:comp}.

\subsection{Comparison with SIS model}\label{subsec:SIS}

We now provide a number of empirical results in which we compare our model, with both finite and infinite memory, to the traditional discrete time SIS model~\cite{YW-DC-CW-CF:03}. In the SIS model, the parameter $ \dsis $ denotes the probability that a node will recover from infection, and $ \bsis $ is the probability that a node will become infected through contact with a single infected neighbour. The dynamics are described through the probability that any node $ i $ will be infected at time $ t $, $ P_i(t) $, which evolves according to the equation
\begin{align*}\label{eq:SIS}
	&P_i(t+1) \\
	=& P_i(t)(1 - \dsis) + (1 - P_i(t))\Big(1 - \prod_{j\in\N_i}(1 - \bsis P_j(t))\Big).\quad
\end{align*}
Note in particular that this model exhibits Markovian behaviour, since the evolution of the process depends only on the probability of infection from the previous time step. We make the simplifying assumption that $ \dsis $ and $ \bsis $ remain the same for all nodes and throughout time, and hence we will compare it with the network Polya contagion process when $ \Delta_r $ and $ \Delta_b $ are similarly fixed in time and throughout the network.

The concept of an epidemic threshold for the SIS model gives a value through which one may determine whether the epidemic dies, a priori using only the system parameters~\cite{YW-DC-CW-CF:03}. The threshold condition is directly related to the largest-magnitude eigenvalue $ \lambda_{max} $ of the adjacency matrix of the underlying graph of the network, and states that if $ \dsis > \bsis\lambda_{max} $ then the epidemic will be eliminated after some time $ n $, i.e., eventually $ P_i(t) = 0 $ for all $ i $ and all $ t > n $. Furthermore, it has been shown that this threshold is tight, and indeed if $ \dsis < \bsis\lambda_{max} $ then some non-zero convergence point exists, called an \emph{endemic state}, and the epidemic will never be eliminated~\cite{HA-BH:13}.

\begin{figure*}[!ht]
\centering
{
	\psfrag{500}[cc][cb]{\scriptsize{$ 500 $}}
	\psfrag{1000}[cc][cb]{\scriptsize{$ 1,000 $}}
	\psfrag{0.475}[rr][rr]{\scriptsize{$ 0.475 $}}
	\psfrag{0.75}[rr][rr]{\scriptsize{$ 0.75 $}}
	\psfrag{0.9}[rr][rr]{\scriptsize{$ 0.9 $}}
	\psfrag{1}[rr][rr]{\scriptsize{$ 1 $}}
	\psfrag{SIS average}[lc][lc]{\scriptsize{SIS average}}
	\psfrag{Memory 50 I}[lc][lc]{\scriptsize{Memory $ 50 $ $ \tilde{I}_n $}}
	\psfrag{Infinite memory I}[lc][lc]{\scriptsize{Inf. memory $ \tilde{I}_n $}}
	\psfrag{Average r}[lc][lc]{\scriptsize{Average $ \rho $}}
	\subfigure[$ \frac{\dsis}{\bsis} = \frac{\Delta_b}{\Delta_r} = \frac{\lambda_{max}}{10} $.]{ \includegraphics[width=0.44\linewidth, height=0.2\textheight]{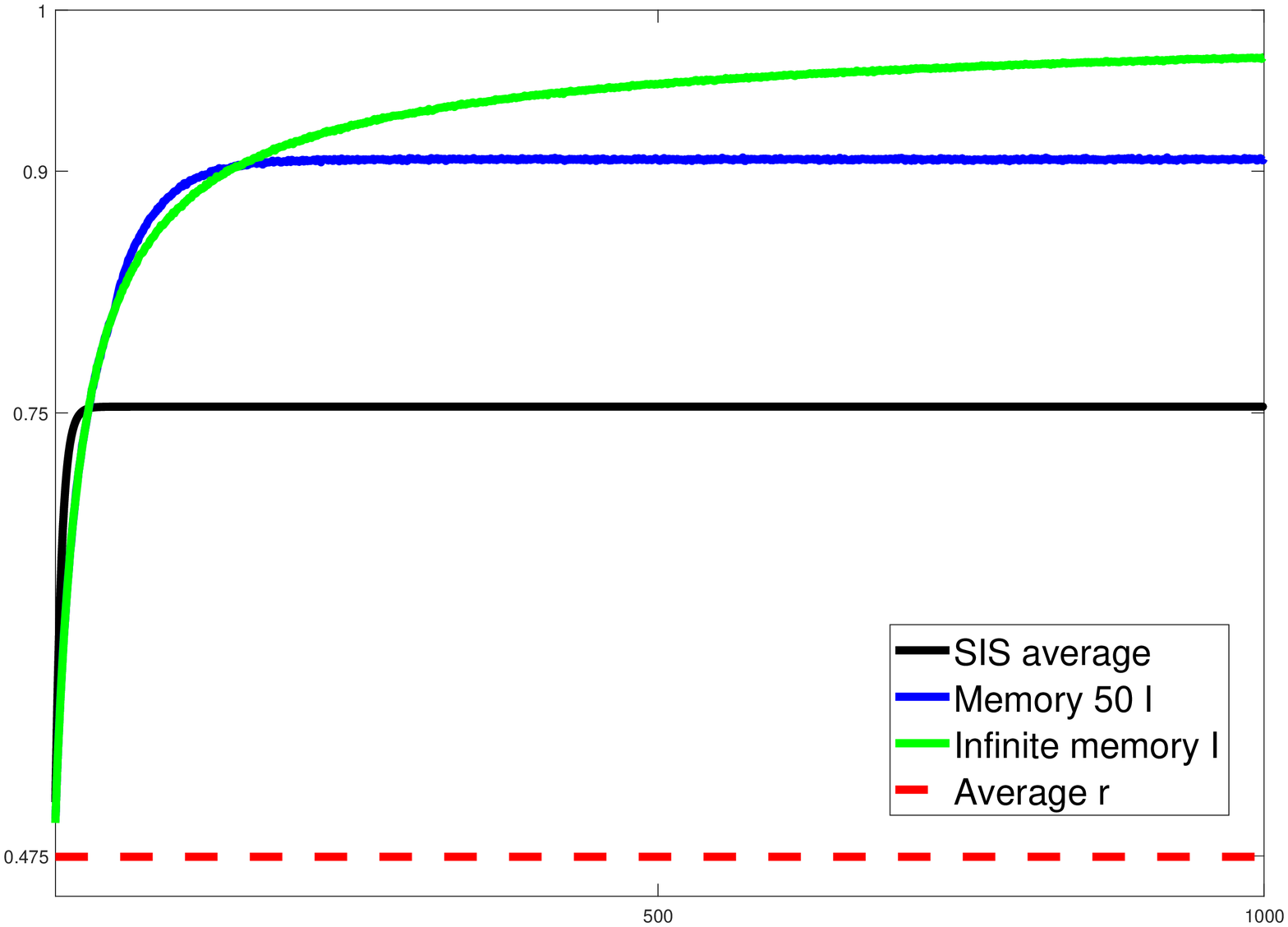} \label{subfig:SIS_threshold_low}}
}
\qquad
{
	\psfrag{500}[cc][cb]{\scriptsize{$ 500 $}}
	\psfrag{1000}[cc][cb]{\scriptsize{$ 1,000 $}}
	\psfrag{0.025}[rr][rr]{\scriptsize{$ 0.025 $}}
	\psfrag{0.475}[rr][rr]{\scriptsize{$ 0.475 $}}
	\psfrag{SIS average}[lc][lc]{\scriptsize{SIS average}}
	\psfrag{Memory 50 I}[lc][lc]{\scriptsize{Memory $ 50 $ $ \tilde{I}_n $}}
	\psfrag{Infinite memory I}[lc][lc]{\scriptsize{Inf. memory $ \tilde{I}_n $}}
	\psfrag{Average r}[lc][lc]{\scriptsize{Average $ \rho $}}
	\subfigure[$ \frac{\dsis}{\bsis} = \frac{\Delta_b}{\Delta_r} = 1.01\lambda_{max} $.]{ \includegraphics[width=0.44\linewidth, height=0.2\textheight]{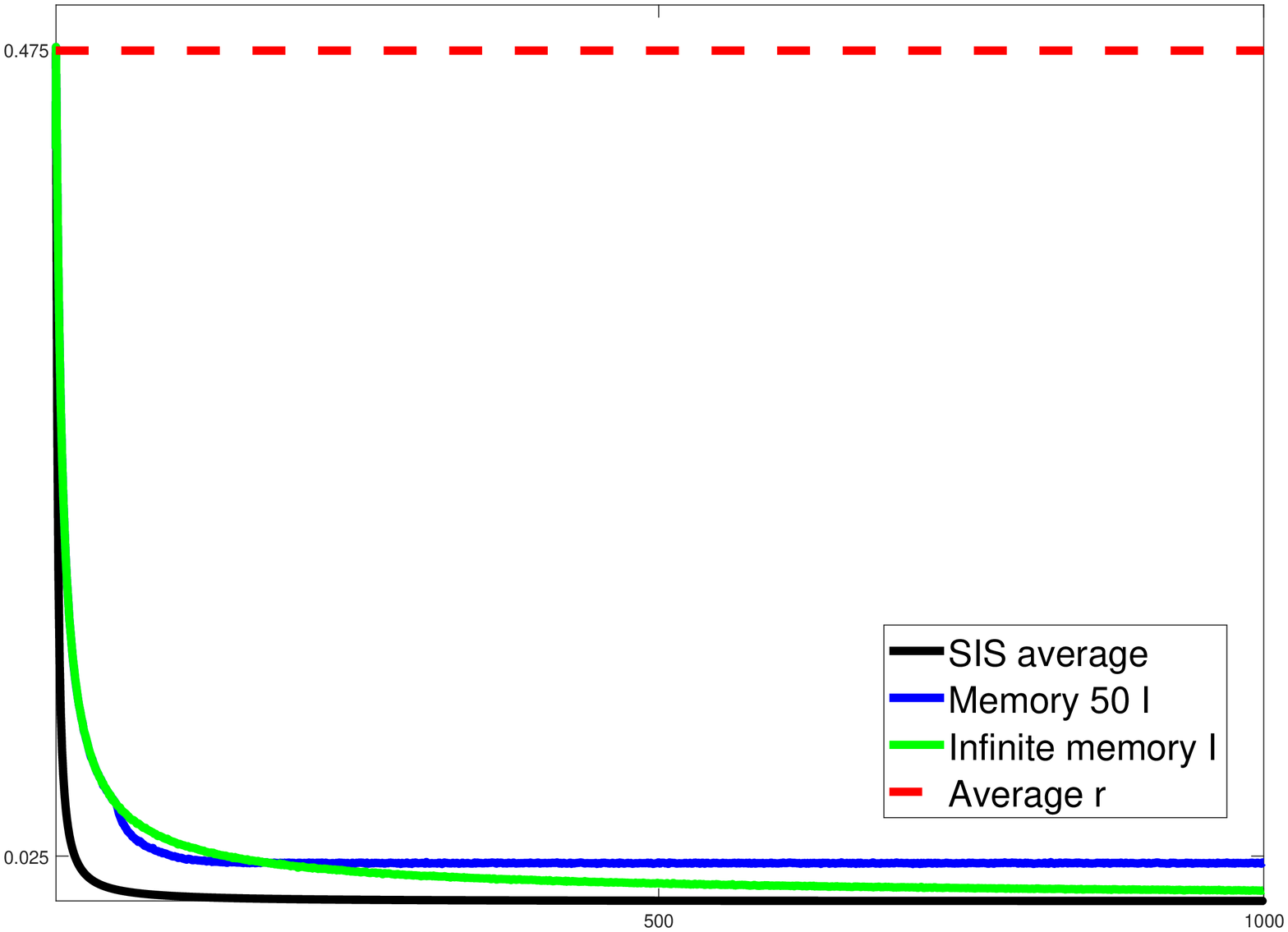} \label{subfig:SIS_threshold_met}}
}
\\
{
	\psfrag{500}[cc][cb]{\scriptsize{$ 500 $}}
	\psfrag{1000}[cc][cb]{\scriptsize{$ 1,000 $}}
	\psfrag{0.469}[rr][rr]{\scriptsize{$ 0.469 $}}
	\psfrag{0.475}[rr][rr]{\scriptsize{$ 0.475 $}}
	\psfrag{SIS average}[lc][lc]{\scriptsize{SIS average}}
	\psfrag{Memory 50 I}[lc][lc]{\scriptsize{Memory $ 50 $ $ \tilde{I}_n $}}
	\psfrag{Infinite memory I}[lc][lc]{\scriptsize{Inf. memory $ \tilde{I}_n $}}
	\psfrag{Average r}[lc][lc]{\scriptsize{Average $ \rho $}}
	\subfigure[$ \frac{\dsis}{\bsis} = \frac{\Delta_b}{\Delta_r} = 1 $ and $ \lambda_{max} > 1 $.]{ \includegraphics[width=0.44\linewidth, height=0.2\textheight]{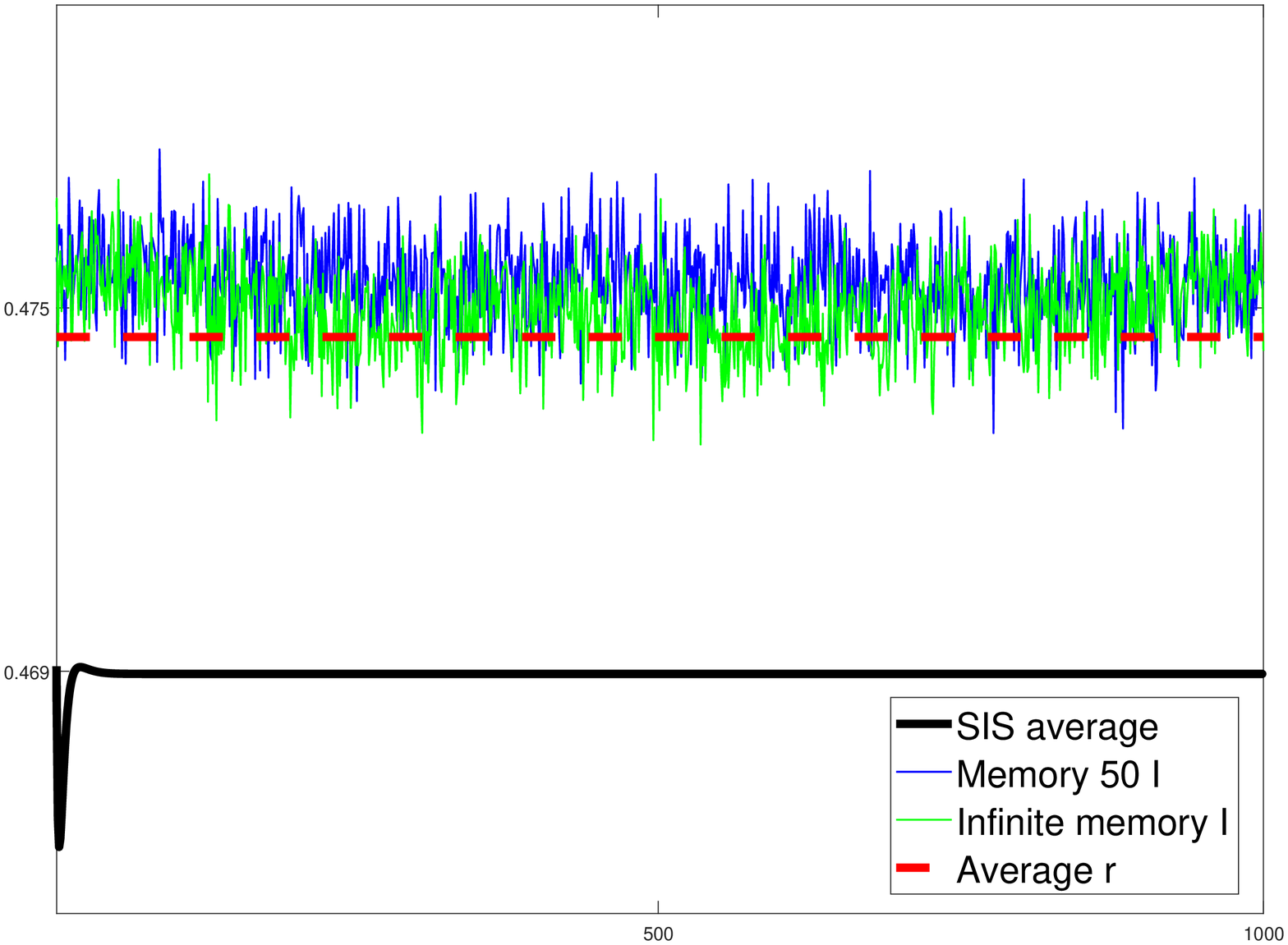} \label{subfig:SIS_ratio_same}}
}
\qquad
{
	\subfigure[100-node Barabasi-Albert network~\cite{RA-ALB:02}.]{ \includegraphics[width=0.44\linewidth, height=0.2\textheight]{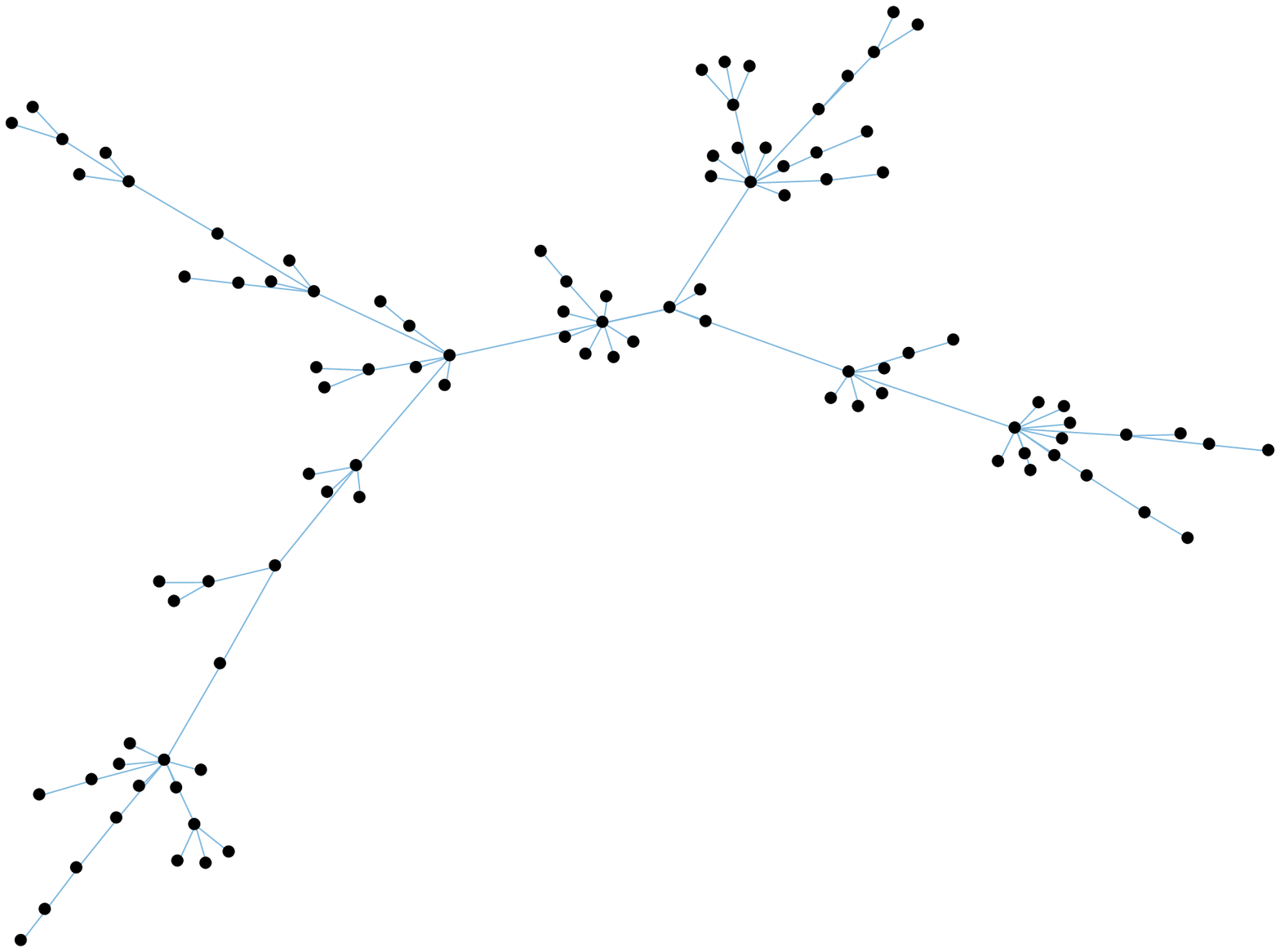} \label{subfig:SIS_network}}
}
\caption[Comparison with SIS model]{Comparison between discrete time SIS model average infection rate $ \frac{1}{N}\Nsum P_i(n) $ and the network Polya contagion process average infection rate $ \tilde{I}_n $. Simulation results were averaged over $5,000$ trials, and the initial parameters $R_i$ and $B_i$ for all nodes were randomly assigned but consistent throughout all trials for a given case. Here $ \lambda_{max} \approx 5.05 $, $ \bsis = 0.15 $ and $ \Delta_r = 2 $ for all cases, while $ \dsis $ and $ \Delta_b $ were set according to the ratios given above.}
\label{fig:SIS}
\vspace{-2.5mm}
\end{figure*}

Figure~\ref{fig:SIS} compares the behaviour of the SIS model and the network Polya contagion process for different selections of these parameters. The initial probabilities of infection $ P_i(0) $ for the SIS model were set to coincide with the initial individual proportions of red balls for the nodes $ \frac{R_i}{T_i} $. Further, we relate in Figures~\ref{subfig:SIS_threshold_low}--\ref{subfig:SIS_ratio_same} the parameters $ \bsis $ and $ \dsis $ to $ \Delta_r $ and $ \Delta_b $, respectively, using ratios of the largest-magnitude eigenvalue $ \lambda_{max} $ of the adjacency matrix of the graph shown in Figure~\ref{subfig:SIS_network}.

Figure~\ref{subfig:SIS_threshold_low} shows a comparison when the SIS model is displaying endemic behaviour. We see here that after a very short time, the SIS model settles and shortly thereafter the finite memory process settles (albeit to a different value), while for the infinite memory process the individual rates of infection and hence the average $ \tilde{I}_n $ continue to increase in time. Since both the SIS model and the finite memory process have limited reinforcement, while the infinite memory process does not, these results are to be expected. Figure~\ref{subfig:SIS_threshold_met} displays a comparison where the epidemic threshold is met and the epidemic dies out for the SIS model. Here we note that $ \tilde{I}_n $ for both the infinite and finite memory processes decreases and approaches zero, albeit not as quickly as the SIS model. Hence we observe that when the curing parameter $ \Delta_b $ is much larger (in fact, more than five folds larger) than the infection parameter $ \Delta_r $ the epidemic is eliminated, as we expect, and this behaviour of the SIS model is captured by the network Polya contagion process. However, the finite memory process does not fully approach zero, since the initial conditions $ R_i $ and $ B_i $ have a much larger influence relative to the infinite memory process. Finally, Figure~\ref{subfig:SIS_ratio_same} shows the case where the epidemic does not vanish and the parameters in both models are set to be equal ($ \dsis = \bsis $ and $ \Delta_b = \Delta_r $). We observe a similar trend between all models, with the finite and infinite memory processes exhibiting near-identical behaviour.

Through these observations, we may conclude that both versions of the network Polya contagion process may apply to the modelling of epidemics, albeit in different applications. The finite memory process exhibits behaviour that is more closely related to the SIS model since they are both limited reinforcement processes, and hence it may be best suited to traditional biological diseases. The infinite memory process obeys similar trends, but in the endemic state there are some interesting differences since the effects of the infection continue to spread throughout the population. On the other hand, the SIS model quickly settles and does not change in time. Thus with infinite memory our process is better suited to modelling opinion dynamics, the spread of ideas, and advertising schemes.

\color{black}

\section{Conclusion}\label{sec:conclusion}

We introduced a network epidemics model based on the classical Polya urn scheme, and we investigated its stochastic properties and asymptotic behaviour in detail. We showed that under certain conditions the proportion of red balls in individual urns and the network susceptibility, which are processes used to measure infection, admit limits. Three classical Polya processes were proposed, one computational and two analytical, to statistically approximate the contagion process of each node. Empirical results were presented which show that the approximations are a good fit for a range of system parameters. \bah{Our process was also compared empirically with the discrete-time SIS model, showing a similar behaviour, particularly in the finite memory mode, while providing different degrees of reinforcement in the endemic state, with the largest reinforcement occurring under the infinite memory mode.} Future directions of research include investigations into the curing of these processes, and the further study of the network contagion process with finite memory.

\bibliographystyle{ieeetr}
\bibliography{alias,Main-add,MH-add}%

\vspace*{-2.2\baselineskip plus -1fil}
\begin{IEEEbiography}[{\includegraphics[width=1in,height=1.25in,clip,keepaspectratio]{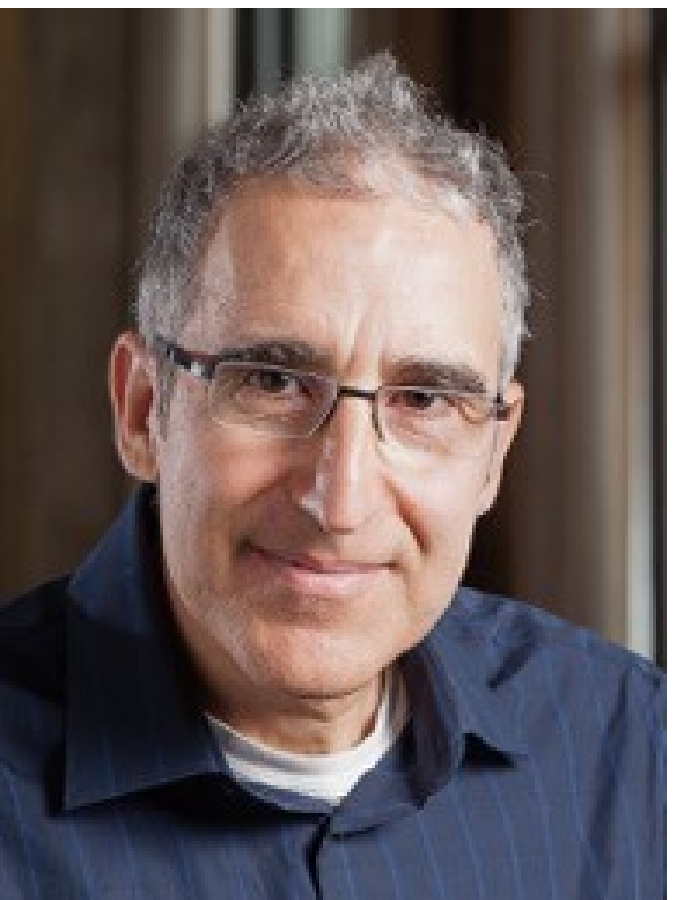}}]{Fady Alajaji}{\relax}
(S'90 - M'94 - SM'00) received the B.E. degree with distinction from the American University of Beirut, Lebanon, and the M.Sc. and Ph.D. degrees from the University of Maryland, College Park, all in electrical engineering, in 1988, 1990 and 1994, respectively. He then held a postdoctoral appointment in 1994 at the Institute for Systems Research, University of Maryland.

In 1995, he joined the Department of Mathematics and Statistics at Queen's University, Kingston, Ontario, where he is currently a Professor of Mathematics and Engineering. Since 1997, he has also held a cross-appointment with the Department of Electrical and Computer Engineering at the same university. In 2013-2014, he served as acting head of the Department of Mathematics and Statistics, and from 2003 to 2008, he served as chair of the Queen's Mathematics and Engineering program. His research interests include information theory, joint source-channel coding, error control coding, data compression and digital communications.

Dr. Alajaji served as Area Editor for Source-Channel Coding and Signal Processing from 2008 to 2015 and as Editor for Source and Source/Channel Coding from 2003 to 2012 for the IEEE Transactions on Communications. He served as organizer and Technical Program Committee member for several international conferences and workshops. He received the Premier's Research Excellence Award from the Province of Ontario in recognition for his research in ``the theory and practice of joint source-channel coding in telecommunication systems.''
\end{IEEEbiography}%

\begin{IEEEbiography}[{\includegraphics[width=1in,height=1.25in,trim={0.6in 0 0.6in 0},clip,keepaspectratio]{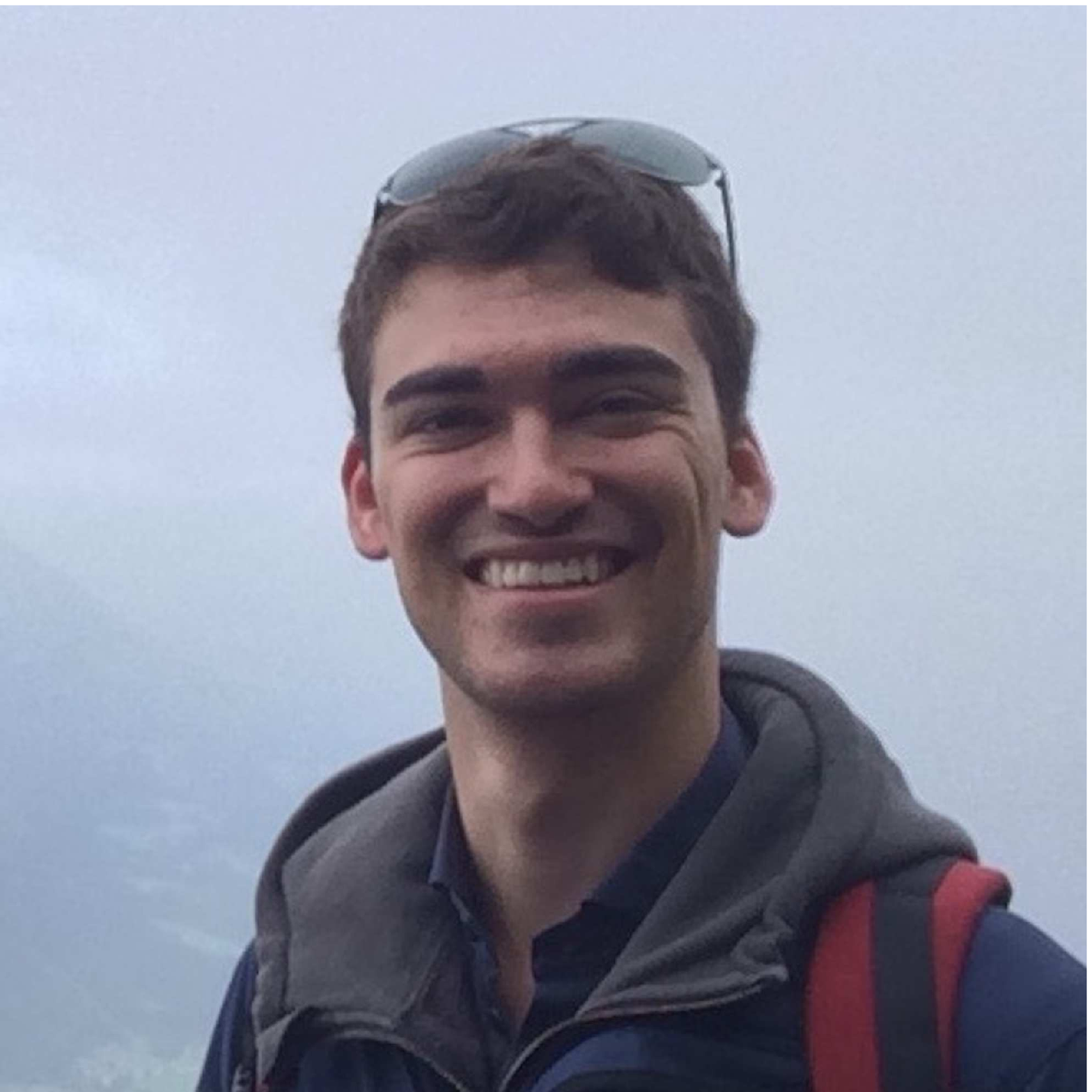}}]{Mikhail Hayhoe}
received the B.Sc. degree in Mathematics and Engineering in 2015 and the M.A.Sc. degree in Applied Mathematics in 2017, both from Queen’s University, Kingston, ON, Canada.

He is currently a Ph.D. student with the Department of Electrical and Systems Engineering at the University of Pennsylvania, Philadelphia, PA, USA. He was a finalist for the Student Best Paper Award at ACC 2017. His research interests include systems and control, optimization, social networks, information theory, big data, and machine learning.
\end{IEEEbiography}%

\vspace*{-2\baselineskip plus -1fil}

\begin{IEEEbiography}[{\includegraphics[width=1in,height=1.25in,clip,keepaspectratio]{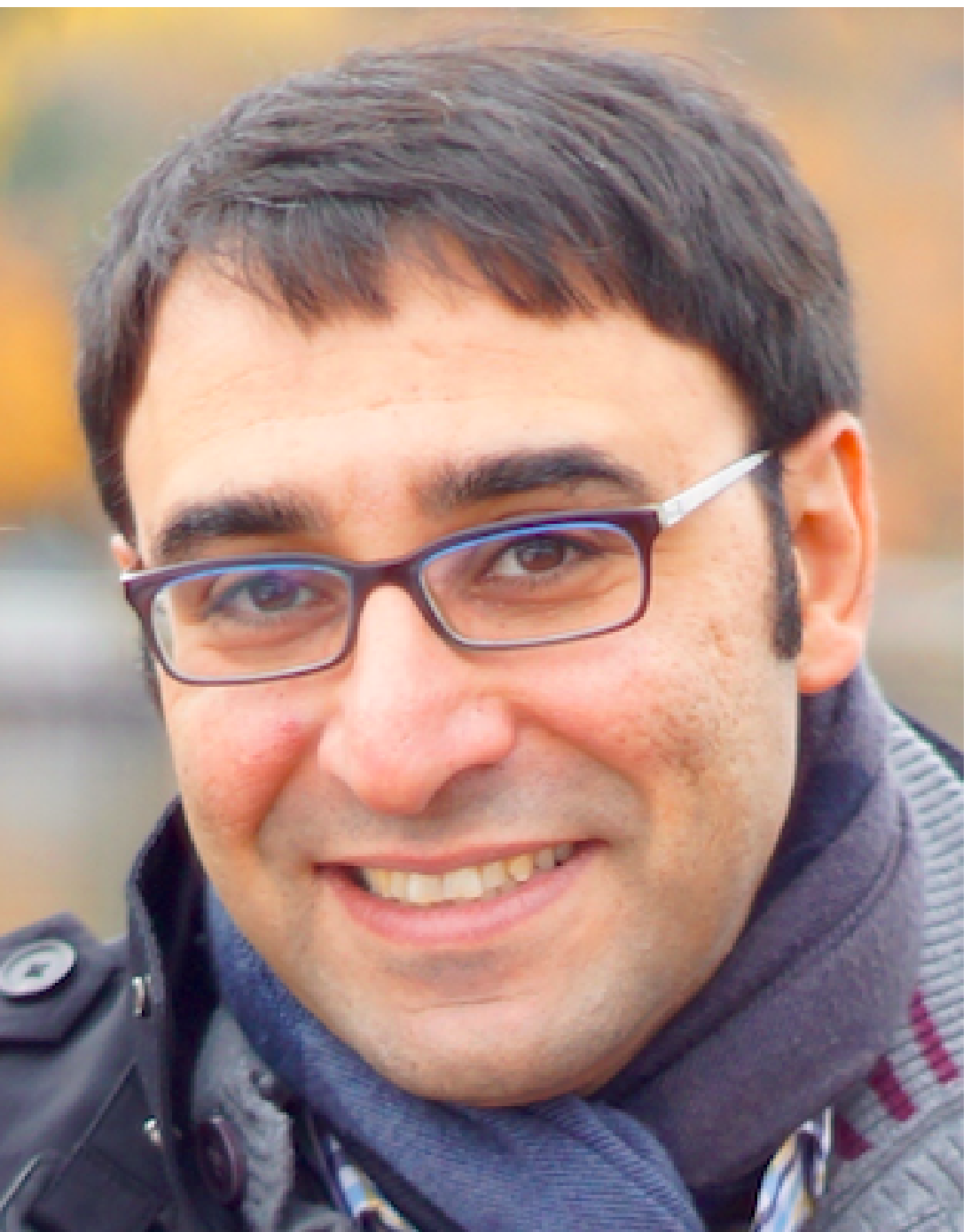}}]{Bahman Gharesifard}
(M'13) received the B.Sc. degree in Mechanical Engineering, in 2002, and the M.Sc. degree in Control and Dynamics, in 2005, from Shiraz University, Iran.  He received the Ph.D. degree in Mathematics, in 2009, from Queen's University, Canada. 

He held postdoctoral positions with the Department of Mechanical and Aerospace Engineering at University of California, San Diego 2009-2012 and with the Coordinated Science Laboratory at the University of Illinois at Urbana-Champaign from 2012-2013.  
He held a visiting faculty position at the Institute for Systems Theory and Automatic Control at the University of Stuttgart in summer of 2016. He is currently an Assistant Professor with the Department of Mathematics and Statistics at Queen's University. His research interests include systems and controls, distributed control and optimization, social and economic networks, game theory, geometric control and mechanics, and Riemannian geometry.
\end{IEEEbiography}
\vfill

\end{document}